%% file: Scalable_Defeasible_Logic.tex
\documentclass{new_tlp}
\usepackage{times}
\usepackage{graphicx}
\usepackage{latexsym}
\usepackage{helvet}
\usepackage{courier}
\usepackage{amsmath}
\usepackage{amssymb}
\usepackage{algorithm}
\usepackage{graphicx}
\usepackage{alltt}
\usepackage{verbatim}
\usepackage{url}
\usepackage{multirow}

\newcommand{\ignore}[1]{}
\newcommand{\finish}[1]{{\bf #1}}
\newcommand{\skipit}[1]{{ #1}}

\newcommand{\DL}{{ DL}}

\newcommand{\pd}[1]{+\partial #1}
\newcommand{\ppd}[1]{+\partial_{||} #1}
\newcommand{\mmd}[1]{-\partial_{||} #1}
\newcommand{\pl}{\partial_{||}}

\newcommand{\PD}[1]{+\Delta #1}
\newcommand{\MD}[1]{-\Delta #1}

\newcommand{\supp}{\sigma}   % support
\newcommand{\non}{{\thicksim}}

\newcommand{\false}{{\bf false}}
\renewcommand{\imath}{i}

\newcommand{\nesubset}{\rotatebox[origin=c]{45}{$\mathbf{\subset}$}}
\newcommand{\sesubset}{\rotatebox[origin=c]{-45}{$\mathbf{\subset}$}}

\newtheorem{theorem}{Theorem}
\newtheorem{lemma}[theorem]{Lemma}

\newtheorem{proposition}[theorem]{Proposition}
\newtheorem{corollary}[theorem]{Corollary}
\newtheorem{example}[theorem]{Example}

\usepackage{pgfplots}
\pgfplotsset{compat=1.14}
\usepackage{tikz}
\usepackage[caption=false]{subfig}

\usepackage{algorithmicx}
\usepackage{algpseudocode}

\begin{document}

\title{Rethinking Defeasible Reasoning:\\A Scalable Approach}

\author[M.J. Maher et al]%[M.J. Maher \and I. Tachmazidis \and G. Antoniou \and S. Wade \and  L. Cheng]
{MICHAEL J. MAHER \\
Reasoning Research Institute, Australia  \\  
\email{michael.maher@reasoning.org.au}
\and
ILIAS TACHMAZIDIS, GRIGORIS ANTONIOU, STEPHEN WADE \\
University of Huddersfield, UK \\
\email{\{i.tachmazidis, g.antoniou, s.j.wade\}@hud.ac.uk}
\and
LONG CHENG \\
Dublin City University, Ireland  \\
\email{long.cheng@dcu.ie}
}

\maketitle
\bibliographystyle{acmtrans}

\begin{abstract}
Recent technological advances have led to unprecedented amounts of
generated data that originate from the Web, sensor networks and social
media. Analytics in terms of defeasible reasoning -- for example for decision
making -- could provide
richer knowledge of the underlying domain. Traditionally, defeasible
reasoning has focused on complex knowledge structures 
%of medium size,
over small to medium amounts of data,
but recent research efforts have attempted to parallelize the reasoning
process over theories with large numbers of facts.  Such work has shown
that traditional defeasible logics come with overheads that limit
scalability.  In this work, we design a new logic for defeasible
reasoning, thus ensuring scalability by design. We establish several
properties of the logic, including its relation to existing defeasible
logics. Our experimental results indicate that our approach is indeed
scalable and defeasible reasoning can be applied to billions of facts.

This paper is under consideration in Theory and Practice of Logic Programming (TPLP).
\end{abstract}

  
\section{Introduction}

Recent technological advances have led to unprecedented amounts of generated 
data that originate from the Web, sensor networks and social media. Once this 
data is stored, the challenge becomes developing solutions for efficient processing 
of the vast amounts of data in order to extract additional value.
Analytics in terms of reasoning -- for example, for decision making -- should be performed using rule sets that would 
allow the aggregation, visualization, understanding and %exploitation 
interpretation 
of given 
datasets and their interconnections. Specifically, one should use rules able to
encode inference semantics, as well as commonsense and practical conclusions in 
order to infer new and useful knowledge based on the data. 

Various monotonic logics have been implemented with this large scale of data in mind.
Work includes 
Datalog~\cite{DBLP:conf/lpnmr/LeoneAACCCFFGLC19,DBLP:journals/tplp/CondieDISYZ18,DBLP:conf/inap/Martinez-AngelesDCB13}, 
$\mathcal{E \kern -0.2em L}^+$ \cite{DBLP:conf/esws/MutharajuHML15},
OWL Horst~\cite{DBLP:conf/bigcomp/KimP15,DBLP:journals/ws/UrbaniKMHB12}, 
RDFS~\cite{DBLP:conf/semweb/HeinoP12,crayreasoning,marvin} and Fuzzy 
logics~\cite{zhou-scale-wl4ai2013,fuzzyHorst_guilin_iswc2011},
scaling reasoning up to billions of facts. For a comprehensive overview of existing approaches
on large-scale reasoning, readers are referred to~\cite{DBLP:journals/ker/AntoniouBMPQTUZ18}.

Nevertheless, it should be pointed out that available data often come from heterogeneous sources that are 
not necessarily controlled by the data engineer, and therefore may contain 
imperfect, incomplete or conflicting information. Other reasons for imperfect data may be faults in sensors or the communication infrastructure. It is evident that monotonic reasoning
is not suited for such data processing, which subsequently led to the study of large-scale
nonmonotonic reasoning. In particular, logic programs under the well-founded 
semantics~\cite{WFS} has been addressed~\cite{DBLP:conf/ruleml/TachmazidisA13,DBLP:journals/tplp/TachmazidisAF14}.
However, this semantics can only indirectly address conflicting information.

Defeasible reasoning provides facilities to directly address conflicting information. While computationally simple, it has found numerous applications in the modelling of legal reasoning~\cite{Prakken,GM17}, regulations~\cite{acis,DBLP:journals/ail/IslamG18}, business rules~\cite{DBLP:conf/sigecom/GrosofLC99}, contracts~\cite{DR-CONTRACT}, negotiation~\cite{DR-NEGOTIATE} and business process compliance management~\cite{FCL,BPCsurvey}. 

However, its scalability is still in question.
Propositional defeasible logics can be executed in linear time \cite{linear,TOCL10} 
but that algorithm does not easily support parallelism, 
nor does it easily extend to first-order defeasible logics.
Implementations of defeasible logic for big data have been limited to subsets of the full logic \cite{DBLP:conf/kr/TachmazidisAFK12,DBLP:conf/ecai/TachmazidisAFKM12}.
% implementations of defeasible logic for big data have been  restricted to either unary predicates~\cite{DBLP:conf/kr/TachmazidisAFK12} 
% or predicates of arbitrary arity, under the assumption of stratification~\cite{DBLP:conf/ecai/TachmazidisAFKM12}.
Both approaches have been applied to billions of facts, but neither approach was able to
capture the general case.
A fundamental problem in existing defeasible logics identified in~\cite{DBLP:conf/ecai/TachmazidisAFKM12}, 
is that the notion of provable-failure-to-prove, which is central to these logics,
requires the generation and retention of a prohibitive amount of negative derivation conclusions.
This inhibits the scalability of implementations of such logics.

In this work, we propose a novel approach for defeasible reasoning over large data
by defining a scalable defeasible logic. 
The new inference rules of the logic avoid reliance on provable-failure-to-prove
by building solely on definitely and defeasibly provable conclusions.
With this approach, the new logic provides scalability by design.
The result is a reasoning process
that is comparable in terms of scalability with existing methods for monotonic logics. In the context of this work, a scalable method allows large-scale inference computation by utilizing parallel and distributed settings over big data.
Experimental results highlight the scalability properties of the proposed logic, while showing that our approach
can scale up to 1 billion facts over a real-world case study.

The paper is structured as follows. Section~\ref{sec:defeasible_logic} provides a brief outline 
of defeasible logic. 
Section~\ref{sec:parallel_dr} discusses an existing implementation of defeasible logic over stratified rule sets,
to demonstrate the general process of inferring defeasible conclusions in a distributed setting and the problems that arise.
The new logic is introduced in Section~\ref{sec:sdl}, while Section~\ref{sec:Properties} 
establishes its theoretical properties. 
Section~\ref{sec:Implementation}
describes the implementation and Section~\ref{sec:Experimental_results} the experimental 
evaluation. We conclude in Section~\ref{sec:CONCLUSION}.
Proofs are available in the appendices.

\section{Defeasible Logics}  \label{sec:defeasible_logic}

%We consider only propositional defeasible logics.
A defeasible theory $D$ is a triple $(F,R,>)$ where $F$ is a finite set of facts (literals), $R$ a finite set of rules,
and $>$ a superiority (or priority) relation (a binary acyclic relation) on $R$,
specifying when one rule overrides another, given that both are applicable.
Rules and facts may be labelled, to enable reference to them.
The set of labels is denoted by $\Lambda(D)$.

A rule $r$ consists (a) of its antecedent (or body) $A(r)$ which is a finite set of literals, (b) an arrow, and, (c) its
consequent (or head) $C(r)$ which is a literal. There are three types of rules: strict rules, defeasible rules and
defeaters represented by a respective arrow $\rightarrow$, $\Rightarrow$ and $\leadsto$. Strict rules are rules in
the classical sense: whenever the premises are indisputable (e.g., facts) then so is the conclusion. Defeasible rules
are rules that can be defeated by contrary evidence. Defeaters are rules that cannot be used to draw any conclusions; their only use is to prevent some conclusions.

A literal is a possibly negated predicate symbol applied to a sequence of variables and constants.
We will require that any variable in the head of a rule also occurs in the body,
and that every fact is variable-free,
a property known as \emph{range-restricted}\footnote{
This is not a requirement of defeasible theories, it simply eases discussion and implementation.
It is a common requirement in work on deductive elements of databases,
and is not very restrictive in practice.
}.
Given a fixed finite set of constants,
any rule is equivalent to a finite set of variable-free rules, and
any defeasible theory is equivalent to a variable-free defeasible theory, for the purpose of semantical analysis.
We refer to variable-free defeasible theories, etc as \emph{propositional}, since there is only a syntactic difference between such theories and true propositional defeasible theories.
Consequently, we will formulate definitions and semantical analysis in propositional terms.
However, for computational analyses and implementation we will also address defeasible theories that are not propositional.

Given a set $R$ of rules, we denote the set of all strict rules in $R$ by
$R_{s}$, and the set of strict and defeasible rules in $R$ by $R_{sd}$. $R[q]$ denotes the set of rules in $R$ with consequent $q$.
If $q$ is a literal, $\non q$ denotes the complementary literal (if $q$ is a positive literal $p$ then $\non q$ is $\neg p$;
and if $q$ is $\neg p$, then $\non q$ is $p$).
A \emph{conclusion} takes the forms $+d \: q$ or $-d \: q$, where $q$ is a literal and $d$ is a tag indicating which inference rules were used.
Given a defeasible theory $D$,
$+d \: q$ expresses that $q$ can be proved via inference rule $d$ from $D$,
while $-d \: q$ expresses that it can be established that $q$ cannot be proved from $D$.

\begin{example}    \label{ex:tweety}
To demonstrate defeasible theories,
we consider the representation of the Tweety problem as a defeasible theory.
The defeasible theory $D$ consists of the rules and facts
\[
\begin{array}{rrcl}
r_{1}: & bird(X) & \Rightarrow & \phantom{\neg} fly(X) \\
r_{2}: & penguin(X) & \Rightarrow & \neg fly(X) \\
r_{3}: & penguin(X) & \rightarrow & \phantom{\neg} bird(X) \\
e      : & bird(eddie) &  &  \\
f      : & penguin(tweety) &  &  \\
\end{array}
\]
and a priority relation $r_{2} > r_{1}$.

Here $r_1, r_2, r_3, e, f$ are labels and
$r_3$ is (a reference to) a strict rule, while $r_1$ and $r_2$ are defeasible rules,
and $e$ and $f$ are facts.
Thus $F = \{e, f\}$, $R_s = \{ r_3 \}$ and $R_{sd} = R = \{r_1, r_2, r_3 \}$
and $>$ consists of the single tuple $(r_2, r_1)$.
The rules express that birds usually fly ($r_1$),
penguins usually don’t fly ($r_2$),
and that all penguins are birds ($r_3$).
In addition, the priority of $r_{2}$ over $r_{1}$ expresses that when something is both a bird and a penguin (that is, when both rules can fire) it usually cannot fly
(that is, only $r_{2}$ may fire, it overrules $r_{1}$).
Finally, we are given the facts that $eddie$ is a bird and $tweety$ is a penguin.
\end{example}

As an example of a defeasible logic, in \cite{TOCL01} a defeasible logic now called $\DL(\partial)$ is defined
with the following inference rules, phrased as conditions on proofs\footnote{
Here, 
$D$ is a defeasible theory  $(F,R,>)$,
$q$ is a variable-free literal,
$P$ denotes a proof (a sequence of conclusions constructed by the inference rules), 
$P[1..i]$ denotes the first $i$ elements of $P$,
and $P(i)$ denotes the $i^{th}$ element of $P$.
}
%\vspace{0.2cm}

\smallskip
\smallskip
\noindent\begin{minipage}[t]{.45\textwidth}
\begin{tabbing}
90123456\=7890\=1234\=5678\=9012\=3456\=\kill

$+\Delta)$  If  $P(i+1) = +\Delta q$  then either \\
\hspace{0.2in}  (1)  $q \in F$;  or \\
\hspace{0.2in}  (2)  $\exists r \in R_{s}[q] \  \forall a \in A(r),
+\Delta a \in P[1..i]$.
\end{tabbing}
\end{minipage}
\begin{minipage}[t]{.45\textwidth}
\begin{tabbing}
90123456\=7890\=1234\=5678\=9012\=3456\=\kill

$-\Delta)$  If  $P(i+1) = -\Delta q$  then \\
\hspace{0.2in}  (1)  $q \notin F$,  and \\
\hspace{0.2in}  (2)  $\forall r \in R_{s}[q] \  \exists a \in A(r),
-\Delta a \in P[1..i]. $
\end{tabbing}
\end{minipage}
\smallskip
\smallskip

These two inference rules concern reasoning about definitive information,
involving only strict rules and facts.
They define conventional monotonic inference.
The next rules refer to defeasible reasoning.

\smallskip
\smallskip
\noindent\begin{minipage}[t]{.45\textwidth}
\begin{tabbing}
$+\partial)$  If  $P(i+1) = +\partial q$  then either \\
\hspace{0.2in}  (1)  $+\Delta q \in P[1..i]$; or  \\
\hspace{0.2in}  (2)  The following three conditions all hold. \\
\hspace{0.4in}      (2.1)  $\exists r \in R_{sd}[q] \  \forall a \in A(r)$,  \\
\hspace{1.1in}                                  $+\partial a \in P[1..i]$,  and \\
\hspace{0.4in}      (2.2)  $-\Delta \non q \in P[1..i]$,  and \\
\hspace{0.4in}      (2.3)  $\forall s \in R[\non q]$  either \\
\hspace{0.6in}         (2.3.1)  $\exists a \in A(s),  -\partial a \in P[1..i]$; \\
\hspace{0.6in}         \ \ or \\
\hspace{0.6in}          (2.3.2)  $\exists t \in R_{sd}[q]$  such that \\
\hspace{0.8in}                $\forall a \in A(t),  +\partial a \in P[1..i]$,  and \\
\hspace{0.8in}                $t > s$.
\end{tabbing}
\end{minipage}
\begin{minipage}[t]{.45\textwidth}
\begin{tabbing}
$-\partial)$  If  $P(i+1) = -\partial q$  then \\
\hspace{0.2in}  (1)  $-\Delta q \in P[1..i]$, and \\
\hspace{0.2in}  (2)  either \\
\hspace{0.4in}      (2.1)  $\forall r \in R_{sd}[q] \  \exists a \in A(r)$,  \\
\hspace{1.1in}                                        $-\partial a \in P[1..i]$; or \\
\hspace{0.4in}      (2.2)  $+\Delta \non q \in P[1..i]$; or \\
\hspace{0.4in}      (2.3)  $\exists s \in R[\non q]$  such that \\
\hspace{0.6in}          (2.3.1)  $\forall a \in A(s),  +\partial a \in P[1..i]$, \\
\hspace{0.6in}           \ \ and \\
\hspace{0.6in}          (2.3.2)  $\forall t \in R_{sd}[q]$  either \\
\hspace{0.8in}                $\exists a \in A(t),  -\partial a \in P[1..i]$;  or \\
\hspace{0.8in}                not$(t > s)$.\\
\end{tabbing}
\end{minipage}

$+\partial q$ is a \emph{consequence} of a defeasible theory $D$ if there is a proof containing $+\partial q$.

In the $+\partial$ inference rule,
(1) ensures that any monotonic consequence is also a defeasible consequence.
(2) allows the application of a rule (2.1) with head $q$, provided that
monotonic inference provably cannot prove $\non q$ (2.2)
and every competing rule either provably fails to apply (2.3.1)
or is overridden by an applicable rule for $q$ (2.3.2).
The $-\partial$ inference rule is the strong negation \cite{flexf} of the $+\partial$ inference rule.
It establishes when a literal is provably not provable in the logic.

To demonstrate these inference rules, we apply them to the Tweety
defeasible theory in the previous example.
\begin{example}   \label{ex:tweety2}
We infer $+\Delta penguin(tweety)$ by application of (1) of the $+\Delta$ inference rule,
and then \linebreak
$+\Delta bird(tweety)$ by application of (2) of that rule.
We also infer $+\Delta bird(eddie)$.
From the $-\Delta$ inference rule we infer
$-\Delta penguin(eddie)$, $-\Delta fly(eddie)$, $-\Delta \neg fly(eddie)$,
$-\Delta fly(tweety)$, and $-\Delta \neg fly(tweety)$,
establishing that these literals cannot be definitely established.

All $+\Delta$ conclusions can also be derived defeasibly, using (1) of the $+\partial$ inference rule.
Also note that $-\partial penguin(eddie)$ is derived because $-\Delta penguin(eddie)$
and there is no (instance of) a rule with head $penguin(eddie)$,
so (1) and (2.1) of the $-\partial$ inference rule are satisfied.
Consequently, we can infer $-\partial \neg fly(eddie)$
because (1) and (2.1)  of the $-\partial$ inference rule
are satisfied by $-\Delta \neg fly(eddie)$ and $-\partial penguin(eddie)$.
Finally, we can now infer $+\partial fly(eddie)$ by (2) of the $+\partial$ inference rule
because $r_1$ and $+\partial bird(eddie)$ combine to satisfy (2.1),
\linebreak
$-\Delta \neg fly(eddie)$ satisfies (2.2),
and (2.3) is satisfied because the only (instance of a) rule for $fly(eddie)$, $r_2$,
has $penguin(eddie)$ in its body, and we derived $-\partial penguin(eddie)$.

In contrast, we infer $+\partial \neg fly(tweety)$ because
$r_2$ and $+\partial penguin(tweety)$ satisfy (2.1),
$-\Delta fly(tweety)$ satisfies (2.2), and
(2.3) is satisfied because the only rule $s = r_1$ with head $\neg fly(tweety)$
is overruled by $t = r_2$ in (2.3.2) using $+\partial penguin(tweety)$.
Without the priority statement, we would not infer $+\partial \neg fly(tweety)$,
and instead infer $-\partial \neg fly(tweety)$ (as well as $-\partial fly(tweety)$),
thus being unable to come to any positive conclusion about the ability of $tweety$ to fly.
\end{example}

A tag/inference rule $d$ in a logic is \emph{consistent} if,
for every defeasible theory $D$ in the logic
and every proposition $q$,
we do not have both consequences
$+d q$
and
$+d\neg q$
unless we also have consequences
$\PD{q}$
and
$\PD{\neg q}$.
This property expresses that defeasible reasoning does not cause inconsistencies:
any inconsistency in consequences is caused by inconsistency in the monotonic part of the defeasible theory. 
We say a logic is consistent if its main inference rule is consistent.
$\DL(\partial)$ is consistent, as are the other logics in \cite{TOCL10}.

\section{Parallel Stratified Defeasible Reasoning}\label{sec:parallel_dr}

In order to facilitate the discussion in the following sections, we first need to 
discuss fundamental notions of parallel stratified defeasible reasoning as presented 
in~\cite{DBLP:conf/ecai/TachmazidisAFKM12}.

A rule set is \emph{stratified} if all of its predicates can be assigned a rank
such that: (a) no predicate depends on one of equal or greater rank, and 
(b) no predicate is assigned a rank not equal to its complement.
Note that a predicate that is found in the head of a rule \emph{depends} on the 
predicates that are found in the body of the same rule.

Consider the following stratified rule set:
\[
\begin{array}{rrcl}
r_{1}: & r(X,Z), s(Z,Y) & \Rightarrow & \phantom{\neg}q(X,Y) \\
r_{2}: & t(X,Z), u(Z,Y)  & \Rightarrow & \neg q(X,Y) \\
       & r_1 > r_2 \\
\end{array}
\]
where predicates $r(X,Z)$, $s(Z,Y)$, $t(X,Z)$ and $u(Z,Y)$ 
are assigned to rank 0, while both $q(X,Y)$ and $\neg q(X,Y)$ are
assigned to rank 1.

Predicates that are assigned to rank 0 do not appear in the head of any rule,
and thus, only a transformation of given facts into $+\Delta$ and $+\partial$ 
conclusions is required. Given the facts $r(a,b)$, $s(b,b)$, $t(a,e)$, $u(e,b)$ and $u(e,g)$, this transformation will create the following conclusions
(assuming a \emph{key}-\emph{value} storage, where the \emph{key} stores the conclusion 
itself while the \emph{value} stores the knowledge about the conclusion):
\begin{center}
	\begin{tabular}{ l l }	
	$<r(a,b),~(+\Delta, +\partial)>$ & ~~~~~ $<s(b,b),~(+\Delta, +\partial)>$ \\
	$<t(a,e),~(+\Delta, +\partial)>$ & ~~~~~ $<u(e,b),~(+\Delta, +\partial)>$ \\
	 & ~~~~~ $<u(e,g),~(+\Delta, +\partial)>$
	\end{tabular}
\end{center}

For rank 1, defeasible reasoning needs to be performed in order to resolve the conflict
between $q(X,Y)$ and $\neg q(X,Y)$. Due to the nature of defeasible reasoning,
parallel reasoning is performed in two passes. The first pass computes applicable rules
for $q(X,Y)$ and $\neg q(X,Y)$. Notice that unlike monotonic reasoning, in 
defeasible reasoning applicable rules might not lead to new conclusions. Hence,
the second pass performs the actual defeasible reasoning and computes for each literal 
whether it is definitely or defeasibly provable. 

The following is based on a distributed system with two nodes. However, the same process
is applicable to any parallel and distributed setting. Note that in order to perform parallel and distributed defeasible reasoning, each node 
requires a complete knowledge of the given rule set, thus enabling both
parallel rule applications (first pass) and parallel defeasible reasoning (second pass).
Consider the following distribution
of the aforementioned conclusions (for the sake of readability, all knowledge is assumed to be 
stored in memory):
\begin{center}
	\begin{tabular}{ l l }
	node 1 &   ~~~~~ node 2\\	
	$<r(a,b),~(+\Delta, +\partial)>$ & ~~~~~ $<s(b,b),~(+\Delta, +\partial)>$ \\
	$<t(a,e),~(+\Delta, +\partial)>$ & ~~~~~ $<u(e,b),~(+\Delta, +\partial)>$ \\
	 & ~~~~~ $<u(e,g),~(+\Delta, +\partial)>$
	\end{tabular}
\end{center}

Considering the first pass, namely computing applicable rules, joins on common arguments
for rule $r_1$ (resp. $r_2$) can only be performed if literals $r(a,b)$ and 
$s(b,b)$ (resp. $t(a,e)$, $u(e,b)$ and $u(e,g)$) are located in the same node, performing joins on argument $b$ (resp. argument $e$). Thus,
the existing knowledge needs to be shuffled as follows:
\begin{center}
	\begin{tabular}{ l l }
	node 1 & ~~~~~ node 2\\	
	$<r(a,b),~(+\Delta,+\partial)>$ & ~~~~~ $<t(a,e),~(+\Delta,+\partial)>$ \\
	$<s(b,b),~(+\Delta,+\partial)>$ & ~~~~~ $<u(e,b),~(+\Delta,+\partial)>$ \\
	 & ~~~~~ $<u(e,g),~(+\Delta,+\partial)>$
	\end{tabular}
\end{center}
Note that \emph{key}-\emph{value} shuffling for the first pass can be performed according to the hash value of the join argument (argument $Z$ in rules $r_1$ and $r_2$), namely after applying a hash function, argument $b$ is assigned to node~1 while argument $e$ is assigned to node~2.
In this way, joins can be performed locally in each node and in parallel since each 
node works independently. Such computation will lead to the following knowledge base:
\begin{center}
	\begin{tabular}{ l l }
	node 1 & ~~~~~ node 2\\	
	$<r(a,b),~(+\Delta,+\partial)>$ & ~~~~~ $<t(a,e),~(+\Delta,+\partial)>$ \\
	$<s(b,b),~(+\Delta,+\partial)>$ & ~~~~~ $<u(e,b),~(+\Delta,+\partial)>$ \\
	$<q(a,b), (+\partial, r_1)>$ & ~~~~~ $<u(e,g),~(+\Delta,+\partial)>$ \\
	 & ~~~~~ $<q(a,b), (\neg,+\partial, r_2)>$	\\
	 & ~~~~~ $<q(a,g), (\neg,+\partial, r_2)>$
	\end{tabular}
\end{center}
where $<q(a,b), (+\partial, r_1)>$ means that $q(a,b)$ is supported by rule $r_1$, $<q(a,b), (\neg,+\partial, r_2)>$ means that $\neg q(a,b)$ is supported by
rule $r_2$, and $<q(a,g), (\neg,+\partial, r_2)>$ means that $\neg q(a,g)$ is supported by
rule $r_2$.

At this point, neither $q(a,b)$ nor $\neg q(a,b)$ can be concluded since the
required knowledge for defeasible reasoning is scattered among different nodes. Thus,
all knowledge for $q(a,b)$ and $\neg q(a,b)$ must be located in a single
node. Hence, the second pass, namely defeasible reasoning, groups all relevant data 
for each potential conclusion in a single node, with different nodes performing reasoning
(in parallel) over different conclusions. Thus, the knowledge will be shuffled as follows:
\begin{center}
	\begin{tabular}{ l l }
	node 1 & ~~~~~ node 2\\	
	$<r(a,b),~(+\Delta,+\partial)>$ & ~~~~~ $<t(a,e),~(+\Delta,+\partial)>$ \\
	$<s(b,b),~(+\Delta,+\partial)>$ & ~~~~~ $<u(e,b),~(+\Delta,+\partial)>$ \\
	$<q(a,b), (+\partial, r_1)>$   & ~~~~~ $<u(e,g),~(+\Delta,+\partial)>$ \\
	$<q(a,b), (\neg,+\partial, r_2)>$	& ~~~~~ $<q(a,g), (\neg,+\partial, r_2)>$
	\end{tabular}
\end{center}

This \emph{key}-\emph{value} shuffling for the second pass can be performed according to the hash value of 
the \emph{key}, namely after applying a hash function, literal $q(a,b)$ is assigned to node~1 while literal $q(a,g)$ is assigned to node~2.
Note that during the second pass only knowledge for literals $q(a,b)$, $\neg q(a,b)$ and $\neg q(a,g)$
is relevant (while literals $r(a,b)$, $s(b,b)$, $t(a,e)$, $u(e,b)$ and $u(e,g)$ are ignored).
Notice that conclusions $q(a,b)$ and $\neg q(a,g)$ are computed in parallel by node 1 and
node 2 respectively.
Finally, after performing defeasible reasoning, the knowledge about applicable rules is replaced by
the final conclusions. Thus, the final knowledge base will contain the following conclusions:
\begin{center}
	\begin{tabular}{ l l }
	node 1 & ~~~~~ node 2\\	
	$<r(a,b),~(+\Delta,+\partial)>$ & ~~~~~ $<t(a,e),~(+\Delta,+\partial)>$ \\
	$<s(b,b),~(+\Delta,+\partial)>$ & ~~~~~ $<u(e,b),~(+\Delta,+\partial)>$ \\
	$<q(a,b), (+\partial)>$           & ~~~~~ $<u(e,g),~(+\Delta,+\partial)>$ \\
	 & ~~~~~ $<\neg q(a,g), (+\partial)>$
	\end{tabular}
\end{center}

For a more elaborate description of parallel stratified defeasible reasoning, readers 
are referred to~\cite{DBLP:phd/ethos/Tachmazidis15}.

Note that in~\cite{DBLP:conf/ecai/TachmazidisAFKM12} only positive conclusions are computed in
order to ensure scalability. In theory, provable-failure-to-prove (e.g., $-\Delta$ and $-\partial$ 
inference rules) could be computed by first calculating applicable rules and then 
applying conflict resolution. However, such computation does not lead to scalable 
solutions.
Consider the following strict rules:
\[
\begin{array}{lrcl}
r^{*}: & p(X,Z),~ q(Z,Y)  & \rightarrow & p(X,Y) \\
r': & q(X,Z), p(Z,V), t(V,Y) & \rightarrow & q(X,Y)
\end{array}
\]

In order to establish that $p(a,b)$ is not definitely provable ($-\Delta$) every possible 
instantiated rule needs to be checked, namely for every value of $Z$ either $p(a,Z)$ 
or $q(Z,b)$ should be established as $-\Delta$. For the aforementioned rule ($r^*$), 
if there are $N$ constants in the given dataset, $N$ instantiated rules need to be checked for each 
$-\Delta$ conclusion.

In general, for $N$ constants in the given dataset and $k$ variables in a given rule that do 
not appear in the head of the rule (e.g., the variable $Z$ in the aforementioned rule $r^{*}$), 
every $-\Delta$ conclusion (say $-\Delta p(X,Y)$) will require $N^{k}$ instantiated rules to be checked. By checking 
every possible instantiated rule, a significant overhead is introduced that can become prohibitive 
even for relatively small datasets (e.g. if $N=10^5$ and $k=3$ then each conclusion would require the computation of $10^{15}$ rules). 

Once all relevant rules
are computed, all available information for each literal (such as $p(a, b)$) must be processed by a single node (containing all relevant information for the literal to be proved). However, this leads to memory and load balancing problems.
For example, if conclusion $-\Delta p(c,b)$ depends on $10^{5}$ instantiated rules (where $N=10^5$ and $k=1$, for rule $r^{*}$), 
while conclusion $-\Delta q(c,b)$ depends on $10^{10}$ instantiated rules (where $N=10^5$ and $k=2$, for rule $r'$), then there 
is a clear difference in the amount of information that needs to be processed by each node 
(a node computing a $-\Delta p(c,b)$ conclusion is expected to terminate significantly faster 
than a node computing a $-\Delta q(c,b)$ conclusion).

%If conclusion $-\Delta p(X,Y)$ depends on $10^{5}$ rules (where $N=10^5$ and $k=1$), while conclusion $-\Delta q(X,Y)$ depends on $10^{10}$ rules (where $N=10^5$ and $k=2$), then there is a clear difference in the amount of information that needs to be processed by each node (a node computing a $-\Delta p(X,Y)$ conclusion is expected to terminate significantly faster than a node computing a $-\Delta q(X,Y)$ conclusion).
%An efficient mechanism for such computation is yet to be defined since it would lead to either main memory insufficiency or skewed load balancing, thus decreasing parallelization. 

This problem motivates the definition of a new logic.

\section{A Scalable Defeasible Logic}\label{sec:sdl}

The defeasible logic $\DL(\pl)$ involves three tags: 
$\Delta$, which we have already seen;
$\lambda$, an auxiliary tag;
and $\pl$, which is the main notion of defeasible proof in this logic.

For a defeasible theory $D$,
we define $P_{\Delta}$ to be 
the set of consequences in the largest proof satisfying the proof condition $+\Delta$,
and call this the $\Delta$ \emph{closure}.
It contains all $+\Delta$ consequences of $D$.

Once $P_{\Delta}$ is computed, we can apply the $+\lambda$ inference rule.
$+\lambda q$ is intended to mean that $q$ is potentially defeasibly provable in $D$.
The $+\lambda$ inference rule is as follows.

\begin{tabbing}
$+\lambda$: \=We may append $P(\imath + 1) = +\lambda q$ if either \\
\> (1) \=$+\Delta q \in P_{\Delta}$ or \\
\> (2)	\>(2.1) $\exists r \in R_{sd}[q] ~ \forall \alpha \in A(r): +\lambda \alpha \in P(1..\imath)$ and \\
\> \>(2.2) $+\Delta \non q \notin P_{\Delta}$ 
\end{tabbing}
\smallskip
%\smallskip

Using this inference rule, and given $P_{\Delta}$, we can compute the $\lambda$ closure $P_{\lambda}$,
which contains all $+\lambda$ consequences of $D$.

$+\pl q$ is intended to mean that $q$ is defeasibly provable in $D$.
Once $P_{\Delta}$ and $P_{\lambda}$ are computed, we can apply the $+\pl$ inference rule.

\begin{tabbing}
$+\pl$: \=We may append $P(\imath + 1) = +\pl q$ if either \\
\> (1) \=$+\Delta q  \in P_{\Delta}$ or \\
\> (2)	\>(2.1) $\exists r \in R_{sd}[q] ~ \forall \alpha \in A(r): +\pl \alpha \in P(1..\imath)$ and \\
\> \>(2.2) $+\Delta \non q \notin P_{\Delta}$ and \\
\> \>(2.3) \=$\forall s \in R[\non q]$ either \\
\> \> \>(2.3.1) $\exists \alpha \in A(s): +\lambda \alpha \notin P_{\lambda}$ or \\
\> \>\>(2.3.\=2) $\exists t \in R_{sd}[q]$ such that \\
\> \>\>\>$\forall \alpha \in A(t): +\pl \alpha \in P(1..\imath)$ and $t > s$
\end{tabbing}

The $\pl$ closure $P_{\pl}$ contains all $\pl$ consequences of $D$.

Notice that the structure of the inference rule for $\pl$ is the same as the structure of the inference rule for $\partial$.
However there are important differences to note:
\begin{itemize}
\item
The inference rule for $\pl$ uses the closures $P_\Delta$ and $P_\lambda$
in addition to the single proof $P$ to which it is applied.
These closures are pre-computed.
In contrast, in $\DL(\partial)$ the proof $P$ incorporates both $\Delta$ and $\partial$ conclusions.
\item
At (2.3.1), the $\pl$ inference rule refers to $\lambda$ rather than $\partial$;
in terms of which conclusions are drawn, this is the most significant variation from $\partial$.
\item
Furthermore,
at (2.2) and (2.3.1), the $\pl$ inference rule does not use negative tags, such as $-\Delta$,
which represent provable failure to prove.  Instead, $\pl$ uses $\notin P$, which represents
failure to prove at the meta level, rather than from within the logic.
This use of $\notin P$ is only possible because it refers to closures that have already been computed.
\end{itemize}

Since the proof rules of our logic do not require $-\lambda$ or $-\pl$ conclusions, we do not present the inference rules for $-\lambda$ and $-\pl$ here.
They are in \ref{app:relexp}.

It is straightforward to see that $\lambda$ is not consistent.
Nevertheless, $\DL(\pl)$ is consistent.
The proof is in \ref{app:ris}.

\begin{proposition}   \label{prop:consistent}
The inference rule $+\pl$ is consistent.
\end{proposition}

Inference rules $\partial$ and $\pl$ employ the notion of ``team defeat'',
where it doesn't matter which rule overrides an opposing rule, as long as all opposing rules are overridden.
This is expressed in (2.3.2).
We can also have a version of $\pl$ with ``individual defeat'',
where all opposing rules must be overridden by the same rule,
which we denote by $\pl^*$.
The inference rule for $+\pl^*$ replaces (2.3.2) in $\pl$ by $r > s$.
It, too, is consistent.

To demonstrate the use of $\DL(\pl)$, we provide a simple example.

\begin{example}   \label{ex:reachable}
Consider the following defeasible theory describing reachability in a directed graph, where some edges may be broken.

\[
\begin{array}{c}

\begin{array}{lrcl}
r: &   reachable(X), link(X, Y)  & \rightarrow & \phantom{\neg} reachable(Y) \\
% & & \rightarrow & reachable(a) \\
s: &  edge(X, Y)  & \Rightarrow & \phantom{\neg} link(X, Y) \\
t: & broken(X, Y) & \Rightarrow & \neg link(X, Y) \\
\end{array} \\
\ \\
\begin{array}{llll}
reachable(a). \\
edge(a, b). & 
edge(b, c). &
edge(b, e). &
edge(c, a). \\
edge(c, d). &
edge(d, e). &
edge(e, d). &
edge(f, e). \\
broken(c, d). &
broken(b, e). \\
\end{array}
\end{array}
\]
with $t > s$.

In this defeasible theory $r$ is a strict rule, $s$ and $t$ are defeasible rules with $t$ overriding $s$ when both are applicable,
and there are facts defining the predicates $edge$ and $broken$, as well as a fact for $reachable$
identifying the starting point for the reachability calculation.

All the facts are known definitely, so they appear in $P_\Delta$; there are no other facts in $P_\Delta$
because there is no definite information about $link$, and so the only strict rule cannot fire.
In addition to all the facts, $P_\lambda$ contains $link(X, Y)$ for each $edge(X, Y)$ fact,
and $\neg link(X, Y)$ for each $broken(X, Y)$ fact. $P_\lambda$ also contains $reachable(X)$ for every $X$ that is reachable from $a$, ignoring the information about broken edges.

$P_{\pl}$ contains all the facts, and $link(X, Y)$ for each unbroken edge and $\neg link(X, Y)$ for each broken edge.
The superiority relation $t > s$ ensures that $+\pl \neg link(X, Y)$ appears and $+\pl link(X, Y)$ does not appear, for each broken edge.
$P_{\pl}$ also contains $reachable(X)$ for every $X$ that is reachable from $a$, via only unbroken edges.
\end{example}

If we compare $\DL(\pl)$ with $\DL(\partial)$ on this defeasible theory
we find that they agree on the defeasible conclusions.
Similarly, on the Tweety theory (Examples \ref{ex:tweety} and \ref{ex:tweety2})
the two logics agree on the defeasible conclusions.

The new inference rules provide a scalability advantage when compared to existing inference rules like $\partial$.
As pointed out in~\cite{DBLP:conf/ecai/TachmazidisAFKM12}, provable-failure-to-prove 
(e.g., $-\Delta$ and $-\partial$ inference rules)
inhibits the scalability of existing defeasible logics. 
We saw a little of this in Example \ref{ex:tweety2},
where many negative conclusions were needed to derive positive conclusions,
but the greater issue arises when rules have variables local to the body of the rule (as discussed in Section~\ref{sec:parallel_dr}).

On the other hand, as illustrated in Section~\ref{sec:Implementation}, the new inference rules that are proposed in this work result in a defeasible
logic that is comparable in terms of scalability to existing monotonic logics, and thus able to benefit
from available optimizations (readers are referred to~\cite{DBLP:journals/ker/AntoniouBMPQTUZ18}
for a comprehensive overview of existing large-scale reasoning methods).

\section{Properties of the Logic}   \label{sec:Properties}

In this section we address properties of $\DL(\pl)$:
the computational complexity of the inference problem,
the relative expressiveness of the logic,
and the relative inference strength of the logic
compared to existing defeasible logics.

\subsection{Computational Complexity}

We first formalize the inference problem for defeasible logics.

\ \\
\noindent
\textbf{The Inference Problem for a Defeasible Logic}

\noindent
\textbf{Instance} \\
A defeasible logic $L$, a defeasible theory $T$, a tag/inference rule $d$ and a literal $q$.

\noindent
\textbf{Question} \\
Is $+d q$ derivable from $T$ using the inference rules of $L$?
\ \\

The computational complexity of inference reflects the difficulty of a scalable implementation.
We show that $\DL(\pl)$ has linear complexity for propositional defeasible theories,
but exponential for arbitrary defeasible theories.
We have three inference rules to consider.
As a result of the structure of the inference rules,
it is straightforward to compute the consequences of $\Delta$, and then $\lambda$, efficiently.

The inference problem for $\partial$ in propositional defeasible logic has linear complexity \cite{linear},
and we use the same techniques to show that
the inference problem for $\pl$ also has linear complexity.
The proof is available in \ref{app:complexity}.

\begin{theorem}   \label{thm:linear}
The set of all consequences of a propositional defeasible theory can be computed in time linear in the size of the defeasible theory.
Consequently, 
the inference problem for propositional $\DL(\pl)$ can be solved in linear time.
\end{theorem}

However, when variables are permitted in rules the inference problem is EXPTIME-complete.

\begin{corollary}   \label{cor:exp}
The set of all consequences of a defeasible theory can be computed in time exponential in the size of the defeasible theory.
Furthermore, the inference problem for defeasible theories is EXPTIME-complete.
\end{corollary}

From a scalability point of view, the potential for parallelism is important.
Unfortunately, the inference problem for $\DL(\pl)$ is not parallelizible in a theoretical sense,
even for propositional defeasible theories.
Inference of $\pl$ consequences of propositional defeasible theories is P-complete, which is generally regarded as a sign that the problem is
not parallelizible (i.e., not computable in poly-log time with polynomially many processors), 
unless all polynomial-time problems are parallelizible.
% model of computation is PRAM: parallel random access machine
% processes access memory in constant time
Actually, inference of $\PD{q}$ is already P-complete,
so all defeasible logics are not parallelizible in this sense. 
But the proof extends to practically every defeasible logic, even without strict rules.

\begin{theorem}   \label{thm:Pcomplete}
The inference problem for propositional defeasible logics is P-complete.
\end{theorem}

The proof is by reduction of the Horn satisfiability problem,
which is P-complete \cite{CookNguyen}.
% The directness and structure of the proof suggest that \emph{any} useful defeasible logic
% will have a P-complete inference problem.

\subsection{Relative Expressiveness}

Relative expressiveness of defeasible logics is defined in terms of the ability
of one logic to simulate another \cite{Maher12,Maher13}, even in the presence of some additions to a theory.
The \emph{addition} of two defeasible theories $(F_1, R_1, >_1) + (F_2, R_2, >_2)$ is
$(F_1 \cup F_2, R_1 \cup R_2, >_1 \cup >_2)$.
Let $\Sigma(D)$ denote the vocabulary of propositions
and $\Lambda(D)$ denote the vocabulary of labels  for $D$.
Given a theory $D$ and a possible simulating theory $D'$,
an addition $A$ is required to be \emph{modular}:
$\Sigma(A) \cap \Sigma(D') \subseteq \Sigma(D)$,
$\Lambda(D) \cap \Lambda(A) = \emptyset$, and
$\Lambda(D') \cap \Lambda(A) = \emptyset$.
This property ensures that the addition $A$ cannot interfere with auxiliary propositions in $D'$,
nor can it interfere by overruling rules in $D$ or $D'$.

A defeasible theory $D_1$ in logic $L_1$ is \emph{simulated} by $D_2$ in $L_2$ with respect to a class $C$ of additions 
if, for every modular addition $A$ in $C$,
$D_1 + A$ and $D_2 + A$ have the same consequences in $\Sigma(D_1 + A)$, modulo tags\footnote{
That is, $D_1$ in $L_1$ might produce $+d_1 q$ while 
$D_2$ in $L_2$ produces $+d_2 q$, due to different inference rules in the different logics,
but the set of literals $q$ that are derived is the same.
}.
The classes $C$ of additions considered in \cite{Maher13} are:
the empty theory, theories consisting only of facts, theories consisting only of rules,
and arbitrary theories.  These represent progressively stronger notions of simulation.

We say a logic $L_1$ can be simulated by a logic $L_2$ with respect to a class $C$
if every theory in $L_1$ can be simulated by some theory in $L_2$ with respect to additions from $C$.
We say $L_2$ is \emph{more (or equal) expressive than} $L_1$ wrt $C$
if $L_1$ can be simulated by $L_2$ with respect to $C$.
$L_2$ is \emph{strictly more expressive than} $L_1$ wrt $C$
if $L_2$ is more expressive than $L_1$ and $L_1$ is not more expressive than $L_2$, wrt $C$.

To see the necessity of the restriction to modular additions we present the following example.
\begin{example}
Consider a conventional defeasible logic $L_1$
that is to be simulated by a similar logic $L_2$ that allows only two literals in the body of a rule.
A theory $D_1$ consisting of a single rule 
\[
\begin{array}{lrcl}
r: & b_1, b_2, b_3 & \Rightarrow & h \phantom{p} \\
\end{array}
\]
in $L_1$ might be represented as $D_2$:
\[
\begin{array}{lrcl}
s: & b_1, b_2 \phantom{b_3} & \Rightarrow & tmp \\
r: & tmp, b_3 \phantom{b_3} & \Rightarrow & h \\
\end{array}
\]
in $L_2$.
However, if an addition $A$ were permitted to include the fact $tmp$ (and $b_3$) then
$D_1 + A$ cannot infer $h$, but $D_2 + A$ can infer $h$.
Similarly, if $A$ contains facts $b_1, b_2, b_3$ and
\[
\begin{array}{lrcl}
s': & \phantom{b_1, b_2, b_3, i} & \Rightarrow & \neg tmp \\
\end{array}
\]
with $s' > s$
then 
$D_1 + A$ can infer $h$, but $D_2 + A$ cannot, because $r'$ overrules $r$.
In either case, $L_2$ does not simulate $L_1$, despite the close similarity
of the two logics.
\end{example}
Thus if non-modular additions were permitted, only simulations that do not use 
auxiliary predicates and labels are possible,
and the notion of relative expressiveness would be useless.

In this section we investigate the relative expressiveness of $\DL(\pl)$.
We first show that 
every theory in $\DL(\partial)$ can be simulated in $\DL(\pl)$.
On the other hand, 
$\DL(\partial)$ cannot simulate $\DL(\pl)$.
In fact, there is a single defeasible theory
whose behaviour in $\DL(\pl)$ cannot be simulated in $\DL(\partial)$.
Thus $\DL(\partial)$ is less expressive than $\DL(\pl)$.

\begin{theorem}   \label{thm:pl_v_dl}
$\DL(\partial)$ is strictly less expressive than $\DL(\pl)$ when there are no additions.
More specifically, 
\begin{itemize}
\item
every defeasible theory in $\DL(\partial)$ can be simulated by a defeasible theory in $\DL(\pl)$
\item
there is a defeasible theory $D$ whose consequences in $\DL(\pl)$
cannot be expressed by any defeasible theory in $\DL(\partial)$
\end{itemize}
\end{theorem}

The argument for the second part is based on the following defeasible theory $D$:
\[
\begin{array}{lrcl}
        &             & \Rightarrow & \phantom{\neg} p \\
        & \neg p  & \rightarrow &  \neg p \\
\end{array}
\]
with empty superiority relation.

$+\lambda p$ and $+\pl p$ are consequences of $D$, as is $-\Delta p$, while $-\Delta \neg p $ is not a consequence.  However, there is no defeasible theory $D'$ in which
$+\partial p$ and $-\Delta p$ are consequences but $-\Delta \neg p $ is not.
See the proof in \ref{app:relexp} for details.
The argument for the second part applies equally to logics $\DL(\partial^*)$, $\DL(\delta)$, $\DL(\delta^*)$
(defined in \cite{TOCL10})
because their inference conditions all have the structure that is used in the proof.
% That structure being that, for a conclusion to be drawn, 
% either (1) $+\Delta p$ or (2) $-\Delta \neg p$ must be consequences.
Essentially, this result arises from the fact that inference in  $\DL(\pl)$ uses, for (2.2), the condition $+\Delta \non q \notin P_\Delta$,
whereas the usual defeasible logics use $-\Delta \non q \in P_\Delta$.
Thus all these logics are not more expressive than $\DL(\pl)$, under any kind of addition.

This theorem suggests that defeasible theories in $\DL(\partial)$ (and other logics) could be transformed into theories of $\DL(\pl)$,
and then executed more scalably.
However, the proof does not provide such a transformation.
Furthermore, the overhead of such a transformation and the expansion in size of the theory
could negate the scalability advantages.
Nevertheless, this remains an avenue for future research.

Although $\DL(\pl)$ is able to simulate $\DL(\partial)$ when there are no additions,
it is unable to achieve a simulation when rules can be added.

\begin{theorem}   \label{thm:dl_v_pl}
$\DL(\pl)$ is not more expressive than $\DL(\partial)$
with respect to addition of rules.
\end{theorem}

The question of whether
$\DL(\pl)$ can simulate $\DL(\partial)$ wrt addition of facts remains open.

\subsection{Relative Inference Strength}

We compare the inference strength of the new inference rules to
the rules of existing defeasible logics.
We write $d_1 \subseteq d_2$ if,
for every defeasible theory $T$ and literal $q$,
if $+d_1 q$ is inferred from $T$ then also $+d_2 q$ is inferred from $T$.
This expresses that $d_2$ has greater inference strength than $d_1$,
in the sense that any literal $d_1$ can infer can also be inferred by $d_2$.
We can also view this inclusion as saying that $d_1$ is an under-approximation of $d_2$,
or that $d_2$ is an over-approximation of $d_1$.
We write $d_1 \subset d_2$ (i.e., the inclusion is \emph{strict})
if $d_1 \subseteq d_2$ and there is a defeasible theory $T$ and literal $q$ such that
$+d_2 q$ is inferred from $T$ but $+d_1 q$ is not.

The relationship between the inference rules introduced in this paper 
and those of other defeasible logics is presented in Figure \ref{fig:inclusion2}.
(We follow the notation of \cite{GM17} for the inference rules $\supp_X$.)
The figure omits $\partial^* \subset \lambda$, 
which is difficult to include in such a two-dimensional representation.
Examples show that all the containments are strict, and no containments are missing.
The proof relies on results available in \ref{app:ris}.

\begin{theorem}   \label{thm:fig2}
The containments illustrated in Figure \ref{fig:inclusion2} hold and are strict.
In addition, $\partial^* \subset \lambda$ holds.
There are no other missing containments in the figure.
\end{theorem}

\begin{figure*}[t]
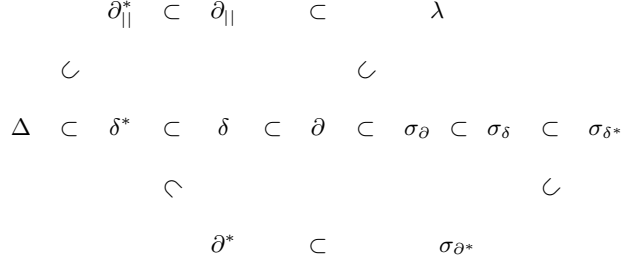

\[
\begin{array}{rccccccccccc}
             &                &  \pl ^*  & \subset & \pl & & \subset & &  \lambda~~~~~~~  \\
\\
             &  \nesubset &                    &   & & & & \nesubset \\
\\
\Delta  & \subset   &    \delta^*          &  \subset  & \delta  &  \subset &  \partial  &  \subset &  \supp_\partial ~~ \subset ~~ \supp_\delta &  \subset &  \supp_{\delta^*} \\
 \\
  & &                                                           &   \sesubset &         & &                &  &                              & \nesubset & \\
 \\
  & &                                                           &                     &      \partial^*  & & \subset & & \supp_{\partial^*}        &\\
\end{array}
\]

\caption{Ordering of defeasible logic inference rules by relative inference strength.
The only unrepresented relation is $\partial^* \subset \lambda$.}
\label{fig:inclusion2}
\end{figure*}

In general, relative inference strength provides an indication of how brave/cautious a logic is in making inferences.
The results show only that $\DL(\pl)$ is incomparable to existing logics.
Nevertheless, the containment $\pl^* \subset \pl$ is noteworthy,
since $\partial^*$ and $\partial$ are incomparable.

\section{Implementation}\label{sec:Implementation}

In this section, we present a generic approach for computing the new inference rules by building
on previous work. Moreover, we outline the implementation for a real-world case study.

\subsection{Import-Apply-Infer}\label{sec:Import-Apply-Infer}
An implementation of the new inference rules should 
first compute the $\Delta$ closure, subsequently the $\lambda$ closure and finally the $\pl$ closure.
It is evident that the $\Delta$ closure computation is conventional rule application,
starting from initial facts and repeatedly applying rules until no new conclusion is derived. Large-scale
closure computation utilizing parallel and distributed settings over big data posses unique challenges, with a wide range of challenges already addressed in the literature for various logics including 
Datalog~\cite{DBLP:journals/tplp/CondieDISYZ18}, $\mathcal{E \kern -0.2em L}^+$~\cite{DBLP:conf/esws/MutharajuHML15},
OWL Horst~\cite{DBLP:conf/bigcomp/KimP15} and RDFS~\cite{DBLP:conf/semweb/HeinoP12}.

\begin{figure}[t]
  \centering
	\includegraphics[width=4.5in]{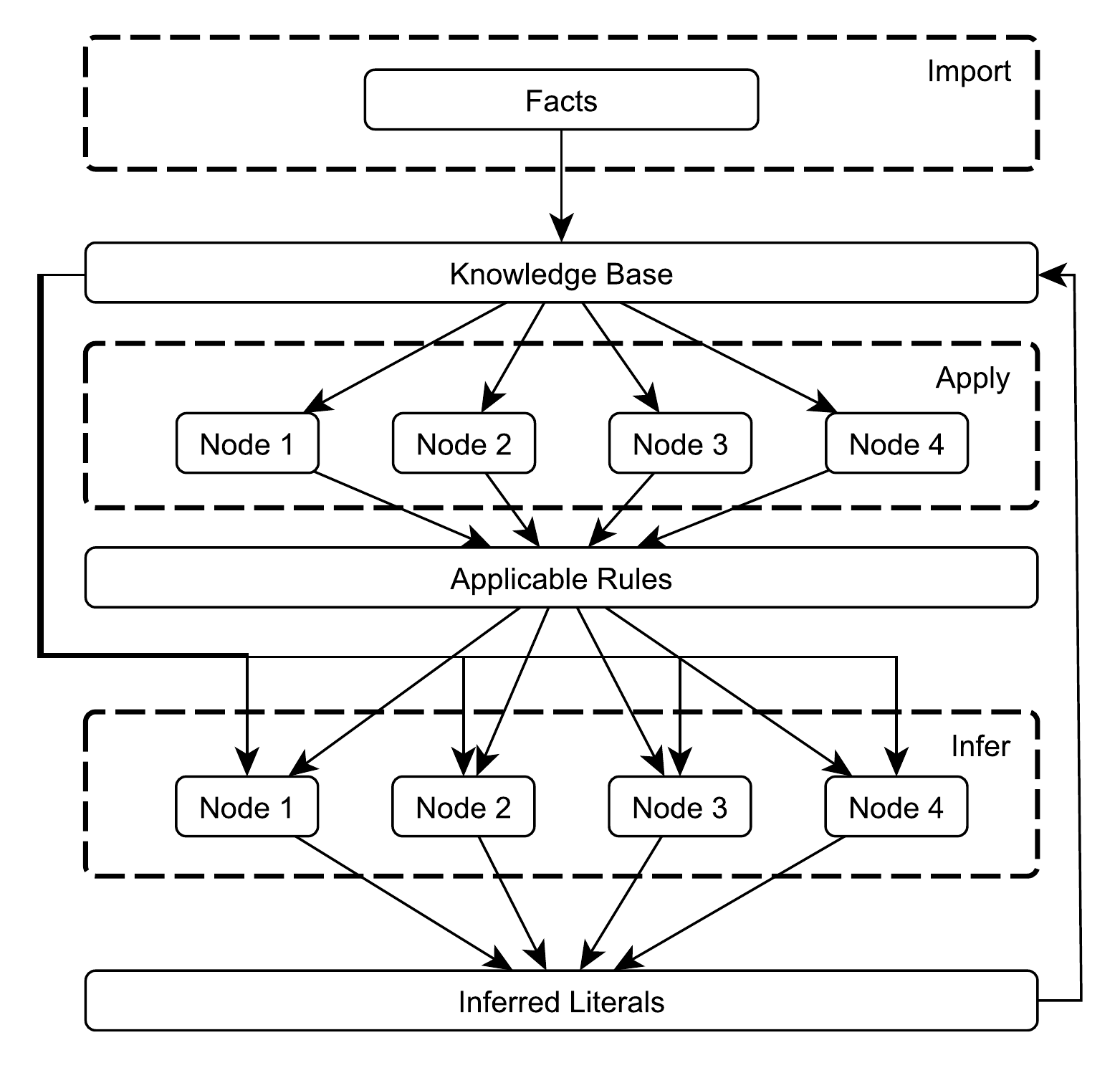}
  \caption{Parallelizing the Import-Apply-Infer model.}
  \label{fig:Import-Apply-Infer}
\end{figure}

For $\lambda$ and $\pl$ inference rules, we propose a three step method called \emph{import-apply-infer},
which can be parallelized as depicted in Figure~\ref{fig:Import-Apply-Infer}. 
Essentially, the first step (\emph{import}) reuses existing knowledge that could be considered as 
facts. Most parallel frameworks provide an efficient data transformation process, thus \emph{import}'s scalability
should be considered self evident.
The second step 
(\emph{apply}) computes all currently applicable rules based on already proved literals. Following
data partitioning, data is divided in chunks with each chunk assigned to a node (4 nodes in Figure~\ref{fig:Import-Apply-Infer}), thus
finding matching literals within each node (e.g., \emph{p(X,Z)} and \emph{q(Z,Y)} match on argument 
\emph{Z} for rule \emph{r*} in Section~\ref{sec:parallel_dr}). Notice that \emph{apply} follows the same 
rule application pattern as the first pass in Section~\ref{sec:parallel_dr}.
The third step (\emph{infer}) resolves existing conflicts (e.g., ``team defeat''),
thus proving and adding new literals to the knowledge base. Notice that \emph{infer} follows the same 
conflict resolution pattern as the second pass in Section~\ref{sec:parallel_dr}. 

Upon close inspection,
$\lambda$ and $\pl$ inference rules are variations (in terms of algorithmic computation) of inference 
rules presented in~\cite{DBLP:conf/ecai/TachmazidisAFKM12}. Thus, the scalability findings 
of~\cite{DBLP:conf/ecai/TachmazidisAFKM12} in terms of
a single computation of steps \emph{import}, \emph{apply} and \emph{infer} are applicable to this work as well.
Considering closure computation,
the \emph{import} step is computed once at the beginning of the process, while steps \emph{apply}
and \emph{infer} are computed repeatedly until no new conclusion is derived. Note that a generic implementation of a parallel reasoner is deferred to future work. 

For the $\lambda$ \emph{closure}, clause (1) of the $\lambda$ inference rule corresponds to the \emph{import} step where 
literals in $P_\Delta$ are treated as given facts, (2.1) of the inference rule corresponds to the 
\emph{apply} step as applicable rules are computed based on already proved $\lambda$ predicates, 
and (2.2) of the inference rule  corresponds to the \emph{infer} step, where a literal $q$ is proved only 
if $+\Delta \non q \notin P_{\Delta}$ (with $P_{\Delta}$ already pre-computed). In terms of scalability,
the \emph{import} step requires importing existing knowledge, which is as scalable as the system's data storage, 
the \emph{apply} step is as scalable as any rule application (including monotonic reasoning), and 
the \emph{infer} step is basic data filtering where knowledge for each literal (both $q$ and $\non q$) 
is processed in parallel by different nodes in the cluster. Note that for any given rule set the 
knowledge for a specific literal is significantly smaller than main memory capacity, while the
large number of literals ensures a high degree of parallelization and scalability.

For the $\pl$ \emph{closure}, clause (1) of the $\pl$ inference rule corresponds to the \emph{import} step where 
literals in $P_\Delta$ are treated as given facts, (2.1) of this inference rule  corresponds to the 
\emph{apply} step as applicable rules are computed based on already proved $\pl$ predicates, 
and clauses (2.2) and (2.3) correspond to the \emph{infer} step, where a literal $q$ 
is proved only if $+\Delta \non q \notin P_{\Delta}$ (with $P_{\Delta}$ already pre-computed), and
either $\exists \alpha \in A(s): +\lambda \alpha \notin P_{\lambda}$ (where $s \in R[\non q]$ 
and $P_{\lambda}$ is already pre-computed) or
$q$ overrides $\non q$ through ``team defeat''. In terms of scalability, $\pl$ \emph{closure}
follows a similar pattern as $\lambda$ \emph{closure} for all three steps. Note that although the \emph{infer} 
step for $\pl$ \emph{closure} requires more complex computations compared to the $\lambda$ \emph{closure},
the amount of processed data for each literal is still significantly smaller than main 
memory capacity, with the large number of literals ensuring a high degree of parallelization and scalability.

\subsection{Apache Spark}\label{sec:Apache_Spark}
We have used Spark\footnote{\url{https://spark.apache.org/}} in our implementations. The main reason is that
the platform is very well suited to parallel data processing in distributed environments. It is elastic in terms of both storage (through the use of HDFS) and computation, which is in contrast with the conventional data systems where each node has to be carefully tuned to its specifications~\cite{cheng2019scalable}. This makes Spark be able to greatly simplify the parallel programming of data applications. Namely, developers only need to focus on the design of high-level workflows and can ignore the underlying parallel executions. To handle the complex workflows in our implementation, we have applied Spark SQL~\cite{armbrust2015spark} in our data processing. Spark SQL is a module in Apache Spark that integrates relational processing with Spark's functional programming API.
Here, we briefly introduce the core abstract of Spark SQL's API - the \textit{DataFrame}.

A DataFrame in Spark SQL is a distributed collection of rows with the same schema. It can be seen as a table in a relational database while its data is distributed over all  computing nodes. A DataFrame can be manipulated and can also perform relational operations over data with existing Spark programs. Currently, DataFrames have supported all the common relational operators, such as projection, filter, join, and aggregations. Moreover, they also enable applications to run SQL queries programmatically and return the result as a DataFrame. Similar to the fundamental data structure of Spark (i.e., RDD), DataFrames are lazy. Namely, in the case that DataFrame object represents a logical plan to compute a dataset, no real parallel execution will occur until an output action such as \textit{save} is called. This mechanism enables Spark SQL to use   data structure information in order to perform rich optimization across all operations that were used to build the DataFrame~\cite{armbrust2015spark}, which is also the main reason why Spark SQL can provide a highly efficient execution solution for data applications.

\subsection{FDA Use Case with Spark}\label{sec:FDA_Use_Case_with_Spark}

The experimental evaluation is based on a FAERS (FDA Adverse Event Reporting System - US Food and Drug Administration) case study,
initially developed for RuleRS~\cite{DBLP:journals/ail/IslamG18}.
More details on this use case are given in Section~\ref{sec:Methodology}.
We have implemented the logic using Spark specifically for this use case,  and made our code, for the evaluated algorithms in this work, publicly 
available\footnote{\url{https://github.com/longcheng11/dReasoning}}. The approach reuses fundamental concepts of \cite{DBLP:conf/ecai/TachmazidisAFKM12}, but is more 
specific to the ruleset, implements the new logic, and uses Spark. Due to the nature of the FDA rule set, reasoning consists mainly of reporting: (a) obligation conclusions (applicable to all FDA cases) that follow from the given rule set, and (b) identified predicates 
for each FDA case. Note that other rule sets 
might require a more elaborate reasoning implementation. 

\begin{algorithm}[t]
\begin{algorithmic}[1]
\State \emph{FDA\_obl\_conclusions} = Seq(``Obligation conclusions'')
\State broadcast(\emph{FDA\_obl\_conclusions})
\State \emph{factsDF} = $\emptyset$
\For{\textbf{each} \emph{inputPath} in \emph{FDA\_Dataset}}
\State \emph{inputDF} = load(\emph{inputPath})
\State \emph{factsDF} += \emph{inputDF}.SQL\_Queries\_To\_Facts()
\EndFor
\State \emph{factsDF} = \emph{factsDF}.Group\_By\_Primary\_ID()
\State \emph{conclusions} = $\emptyset$
\For{\textbf{each} \emph{primaryID} in \emph{factsDF}}
\State \emph{conclusions} += reasoning(\emph{FDA\_obl\_conclusions}, \emph{factsDF}.getFacts(\emph{primaryID}))
\EndFor
\State \emph{conclusions}.count()
\end{algorithmic}
\caption{Spark implementation}
\label{Spark_implementation}
\end{algorithm}

The basic structure of the implementation is described in 
Algorithm~\ref{Spark_implementation}.
First, a set of obligation conclusions, that is a set of obligations
that need to be concluded for all FDA cases  (such as the obligation to report ``Patient age'', i.e., ``obl\_report\_Patient\_age\_to\_FDA''), is loaded in memory (line~1). Note that obligation conclusions are manually extracted from the FDA rule set. In order to allow each node in the cluster to perform reasoning independently by providing all required 
information on the given rule set, the set of obligation
conclusions needs to be broadcast (line~2). 
Prior to applying SQL queries, facts are initially set to an empty DataFrame (line~3). 
Subsequently, each input file in the extracted FDA dataset (lines~4-7) is loaded into a corresponding DataFrame (line~5), and 
SQL queries are executed (line~6), using Spark SQL, in order to extract predicates (facts) that will 
be used for reasoning. Spark SQL ensures parallel evaluation of given SQL queries, 
while the developer needs only to define the queries using the Spark SQL API. Prior to performing 
reasoning, generated facts are grouped based on their \emph{primaryid}, namely each FDA 
case is handled separately (line~8). Data grouping is performed in parallel by Spark.
Note that lines~3-8 should be considered as the \emph{import} step.

Considering the reasoning process itself, conclusions are initially set to an empty DataFrame (line~9), 
while reasoning over each \emph{primaryid} (in parallel) adds new conclusions (lines~10-12). Essentially,
each \emph{primaryid} is evaluated by a different node in the cluster, thus ensuring parallelism.
Note that lines 9-12 should be considered as steps \emph{apply} and \emph{infer}.
Finally, the number of conclusions is counted 
(line~13). Counting the number of final conclusions instead of materialising the
output allows a better focus on the runtime performance of a given implementation.
At the end of Algorithm~\ref{Spark_implementation},  final conclusions are stored in memory  and could be readily used for further processing if required.

\section{Experimental Results}\label{sec:Experimental_results}
In this section, we present the results of our experimental
evaluation on a commodity cluster. We conduct a quantitative
evaluation of our implementation.

%%%%%%%%%%%%%%%%%%%%%
%%% Input details %%%
%%%%%%%%%%%%%%%%%%%%%
\begin{table}[t]
\caption{Input details}
\label{tab:input_details}
\centering
%\small
\begin{tabular}{rrrrr}
	\hline\hline
	Copies & Size (GB) & Distinct cases & Rows & Facts \\ [0.5ex]
	\hline\hline
	1 & 3 & 5,285,699 & 43,791,158 & 96,925,980 \\ \hline
	
	3 & 9 & 15,857,097 & 131,373,474 & 290,777,940 \\ \hline

	6 & 18 & 31,714,194 & 262,746,948 & 581,555,880 \\ \hline

    12 & 36 & 63,428,388 & 525,493,896 & 1,163,111,760 \\ 
	\hline\hline	
\end{tabular}
\end{table}

%%%%%%%%%%%%%%%%%%%%%%%%%%%%%%%
%%% Number of rows per file %%%
%%%%%%%%%%%%%%%%%%%%%%%%%%%%%%%
\begin{table}[t]
\caption{Number of rows per file}
\label{tab:Number_of_rows_per_file}
\centering
%\small
\begin{tabular}{llllll}
	\hline\hline
	DEMO & DRUG & OUTC & REAC & RPSR & Whole \\ [0.5ex]
	\hline\hline
	5,285,792 & 19,087,015 & 3,649,558 & 15,525,084 & 243,709 & 43,791,158   \\ 
	\hline\hline	
\end{tabular}
\end{table}

%%%%%%%%%%%%%%%%%%%%%%%%%%%%%
%%% Number of SQL queries %%%
%%%%%%%%%%%%%%%%%%%%%%%%%%%%%
\begin{table}[t]
\caption{Number of SQL queries}
\label{tab:number_of_SQL_queries}
\centering
%\small
\begin{tabular}{llllll}
	\hline\hline
	DEMO & DRUG & OUTC & REAC & RPSR & Whole \\ [0.5ex]
	\hline\hline
	13 & 5 & 1 & 1 & 1 & 21   \\ 
	\hline\hline	
\end{tabular}
\end{table}

%%%%%%%%%%%%%%%%%%%%%%%%%%%%%
%%% Number of conclusions %%%
%%%%%%%%%%%%%%%%%%%%%%%%%%%%%
\ignore{
\begin{table}[t]
\caption{Number of conclusions}
\label{tab:number_of_conclusions}
\centering
%\small
\begin{tabular}{lllllll}
	\hline\hline
	Copies & DEMO & DRUG & OUTC & REAC & RPSR & Whole \\ [0.5ex]
	\hline\hline
	1 & 310,390,445 & 291,485,006  & 150,875,577 & 280,142,047 & 10,781,366 & 335,354,933  \\ \hline
	
	3 & 931,171,335 & 874,455,018 & 452,626,731 & 840,426,141 & 32,344,098 & 1,006,064,799  \\ \hline

	6 & 1,862,342,670 & 1,748,910,036 & 905,253,462 & 1,680,852,282 & 64,688,196 & 2,012,129,598  \\ \hline

    12 & 3,724,665,360 & 3,497,800,984 & 1,810,491,448 & 3,361,686,226 & 129,376,392 & 4,024,237,482  \\ 
	\hline\hline	
\end{tabular}
\end{table}
}

\begin{table}[t]
\caption{Number of conclusions (in millions)}
\label{tab:number_of_conclusions}
\centering
%\small
\begin{tabular}{rrrrrrr}
	\hline\hline
	Copies & DEMO & DRUG & OUTC & REAC & RPSR & Whole \\ [0.5ex]
	\hline\hline
	1 & 310.390 & 291.485  & 150.876 & 280.142 & 10.781 & 335.355  \\ \hline
	
	3 & 931.171 & 874.455 & 452.627 & 840.426 & 32.344 & 1,006.065  \\ \hline

	6 & 1,862.343 & 1,748.910 & 905.253 & 1,680.852 & 64.688 & 2,012.130  \\ \hline

    12 & 3,724.665 & 3,497.801 & 1,810.491 & 3,361.686 & 129.376 & 4,024.237  \\ 
	\hline\hline	
\end{tabular}
\end{table}

\subsection{Methodology}\label{sec:Methodology}
The evaluation of our approach is based on the RuleRS~\cite{DBLP:journals/ail/IslamG18}
FAERS (FDA Adverse Event Reporting System - US Food and Drug Administration) 
case study. The FDA Adverse Event Reporting System (FAERS) is a database that contains 
adverse event reports, medication error reports and product quality complaints 
resulting in adverse events that were submitted to FDA. The database is designed 
to support the FDA's post-marketing safety surveillance program for drug and 
therapeutic biologic products\footnote{\url{https://www.fda.gov/drugs/surveillance/questions-and-answers-fdas-adverse-}\\\url{event-reporting-system-faers}}.

{\bf Dataset.} 
FAERS publishes quarterly data files\footnote{\url{https://fis.fda.gov/extensions/FPD-QDE-FAERS/FPD-QDE-FAERS.html}}, which include: 
\begin{itemize}
  \item \emph{DEMO}: Demographic and administrative information.
  \item \emph{DRUG}: Drug information from the case reports.
  \item \emph{OUTC}: Patient outcome information from the reports.
  \item \emph{REAC}: Reaction information from the reports.
  \item \emph{RPSR}: Information on the source of the reports.
\end{itemize}

Table~\ref{tab:input_details} describes the
details of the used input. The original dataset consists of data published 
between the third quarter of 2014 and the second quarter of 2018 (a total of four calendar years), which 
corresponds to 3GB of storage space, 5,285,699 distinct FDA cases (with
each case indicated by a unique \emph{primaryid}), 43,791,158 rows in the
consolidated CSV files (for more details see Table~\ref{tab:Number_of_rows_per_file}), 
and 96,925,980 generated facts by the SQL queries 
(when all SQL queries were applied). Note that details of applied SQL queries are 
described below.

The initial dataset allows reasoning over 97M facts, which would not highlight the
full potential of the proposed method. For scalability purposes, copies of the aforementioned 
dataset were generated by adjusting the \emph{primaryid} field, where 3, 6 and 12 copies 
correspond to 291M, 582M and 1.16 billion facts respectively. Note that the \emph{primaryid} field
is adjusted by appending a counter, namely for 3 copies the following input:
\begin{center}
	\begin{tabular}{ l l l }
		primaryid & caseid & rpsr\_cod \\
		100208273 & 10020827 & FGN
	\end{tabular}
\end{center}
would be transformed into:
\begin{center}
	\begin{tabular}{ l l l }
		primaryid & caseid & rpsr\_cod \\
		100208273{\bf1} & 10020827 & FGN \\
		100208273{\bf2} & 10020827 & FGN \\
		100208273{\bf3} & 10020827 & FGN				
	\end{tabular}
\end{center}
Note that the first step from 1 copy to 3 copies is counter-intuitive in terms of scalability
(not a power of two), however it still provides interpretable results while allowing
an evaluation of up to 1.16 billion facts (for 12 copies).

{\bf Rule set.} 
The rule set consist of rules that are manually converted from \emph{U.S. ELECTRONIC CODE OF FEDERAL 
REGULATIONS, Title 21: Food and Drugs, PART 310-NEW DRUGS, Subpart D-Records and Reports US Government 
(2014)}\footnote{\url{https://www.ecfr.gov/cgi-bin/text-idx?SID=7bf64fa0b8f5d9185244a769699c5e13&mc=true&node=se21.5.310_1305&rgn=div8}}.
As discussed in~\cite{DBLP:journals/ail/IslamG18}, the regulations: (a) specify  the records and reports concerning
adverse drug experiences on marketed prescription drugs for human use without
approved new drug applications, and (b) include reporting requirements for Manufacturers, Packers, and 
Distributors (MPD) and information reported on various life-threatening serious and unexpected 
adverse drug experience for Individual Case Safety Report (ICSR). 

Consider for example the provision (as part of informed on ICSRs) prescribing to report electronically 
to FDA as ICSRs to include ``Patient age'' while reporting to FDA. Its formal representation in Defeasible 
Deontic Logic is:
\begin{center}
	\begin{tabular}{ l }
        $r_1$ : \  [OAPNP]\:\emph{report\_on\_ICSRs\_to\_FDA}(X) \  $\Rightarrow$ \  [OAPNP]\:\emph{report\_Patient\_age\_to\_FDA}(X)
	\end{tabular}
\end{center}
where [OAPNP] is a deontic operator expressing obligation,
while using the defeasible logic as defined in this work we have:
\begin{center}
	\begin{tabular}{ l }
        $r_1$ : \emph{obl\_report\_on\_ICSRs\_to\_FDA}(X) \ $\Rightarrow$ \ \emph{obl\_report\_Patient\_age\_to\_FDA}(X)
	\end{tabular}
\end{center}
Note that obligation conclusions such as 
\emph{obl\_report\_Patient\_age\_to\_FDA(X)} in rule $r_1$ are loaded in memory and broadcast
to each node in the implementation in order to ensure parallelism (see lines 1-2 in Algorithm~\ref{Spark_implementation}).

{\bf SQL Queries.} 
In~\cite{DBLP:journals/ail/IslamG18} predicate extraction (facts generation) is performed through SQL queries 
over a PostgreSQL database while in our experiments such SQL queries are part of the implementation 
using Spark SQL. Each query represents a single predicate, which is eventually passed as an input 
parameter to the reasoner. Note that the reasoner is implemented using Spark, based on 
Algorithm~\ref{Spark_implementation}. 

The following query illustrates sample test predicates \emph{ICSRs\_contain\_Patient\_age}
in PostgreSQL:
\begin{verbatim}
SELECT primaryid,
CASE WHEN age IS NOT NULL THEN 'report_Patient_age_to_FDA'          
     ELSE '-report_Patient_age_to_FDA'
END
FROM DEMO14Q1
\end{verbatim}

Such query evaluation can be implemented using Spark SQL by loading in memory a
\linebreak
DataFrame, 
called \emph{demoDF}, containing records found in file \emph{DEMO} where attribute \emph{age}
must be defined in a given row. The aforementioned SQL query can be translated in Spark SQL
as follows (note that in this work only positive literals are relevant):
\begin{verbatim}
demoDF
 .where(demoDF.col("age").isNotNull && demoDF.col("age") =!= "")
 .select(demoDF.col("primaryid").as("argument_X"))
 .withColumn("predicate",lit("report_Patient_age_to_FDA"))
\end{verbatim}

The following input (note that file \emph{DEMO} contains more columns, which are 
not included below for readability purposes):
\begin{center}
	\begin{tabular}{ l l l l l }
		primaryid & caseid & ... & age & ... \\
		100051922 & 10005192 & ... & 21 & ...
	\end{tabular}
\end{center}
would be transformed into:
\begin{center}
	\begin{tabular}{ l l }
		argument\_X & predicate \\
		100051922 & report\_Patient\_age\_to\_FDA 	
	\end{tabular}
\end{center}
which essentially represents the fact \emph{report\_Patient\_age\_to\_FDA(100051922)}.

The number of executed queries for each file is included in  Table~\ref{tab:number_of_SQL_queries}, 
where it is clear that the majority of queries are executed over files \emph{DEMO} and \emph{DRUG}. 
Note that all SQL queries are included in our publicly available implementation 
(see Section~\ref{sec:FDA_Use_Case_with_Spark}).

All rules are transformed into the notation defined in this work with the reasoner being implemented 
in Spark specifically for this rule set.  Both fact extraction (Spark SQL) and reasoning (Spark) are implemented 
within a single job as described in Algorithm~\ref{Spark_implementation} (see Section~\ref{sec:FDA_Use_Case_with_Spark}).

%%%%%%%%%%%%%%%%%%%%%%%%%%%%%%%%%%%%%%%%%%%%%%%%%%%%%%%%%%%%%%%%%%%%%%%%%%%%%%%%%%%%%%%
%%% Time in seconds as a function of number of facts, for various numbers of nodes. %%%
%%%%%%%%%%%%%%%%%%%%%%%%%%%%%%%%%%%%%%%%%%%%%%%%%%%%%%%%%%%%%%%%%%%%%%%%%%%%%%%%%%%%%%%
\begin{figure}[t]
%\begin{subfigure}{.498\linewidth}
\centering
\subfloat{
\begin{tikzpicture}[node distance = 0cm, scale=0.8, transform shape]
\begin{axis}[xlabel=Millions of facts,ylabel= CPU time (s),legend pos= north west]
    \addplot+[color=red,mark=otimes, mark size=2,error bars/.cd, y dir=both,y explicit] coordinates {
	(97,123)
	(291,452)
	(582,1109)
	(1163,2390)
    };
    \addplot+[color=blue,mark=x, mark size=2,error bars/.cd, y dir=both,y explicit] coordinates {
	(97,92)			
	(291,247)
	(582,563)
	(1163,1572)
    };
    \addplot+[color=green,mark=diamond, mark size=2, error bars/.cd, y dir=both,y explicit] coordinates {
	(97,76)
	(291,203)			
	(582,412)
	(1163,796)
    };
    \addplot+[color=black,mark=square, mark size=2, error bars/.cd, y dir=both,y explicit] coordinates {
	(97,77)
	(291,129)			
	(582,231)
	(1163,572)
    };
    \addplot+[color=magenta,mark=triangle, mark size=2, error bars/.cd, y dir=both,y explicit] coordinates {
	(97,83)
	(291,111)			
	(582,169)
	(1163,344)
    };
    \legend{1 node, 2 nodes, 4 nodes, 8 nodes, 16 nodes}
    \end{axis}
\end{tikzpicture}
\label{fig:eval-whole-sub}
}
\caption{Time in seconds as a function of number of facts, for various numbers of nodes.}
\label{fig:eval-whole}
\end{figure}
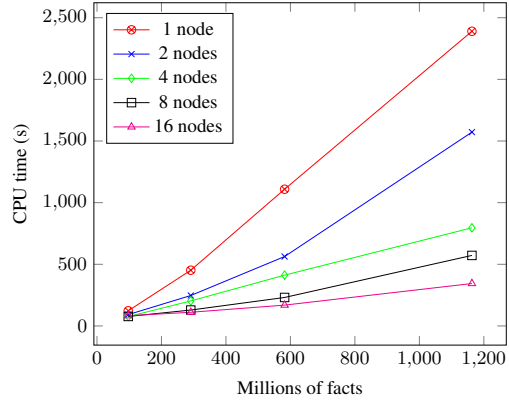

%%%%%%%%%%%%%%%%%%%%%%%%%%%%%%%%%%%%%%%%%%%%%%%%%%%%%%%%%%%%%%%%%%%%%%%%%%%%%%%%%%%%%%%%%%%%%%%%%%%%%
%%% Speed-up and scaled speed-up as a function of numbers of nodes, for various numbers of facts. %%%
%%%%%%%%%%%%%%%%%%%%%%%%%%%%%%%%%%%%%%%%%%%%%%%%%%%%%%%%%%%%%%%%%%%%%%%%%%%%%%%%%%%%%%%%%%%%%%%%%%%%%
\begin{figure}[t]
%\begin{subfigure}{.498\linewidth}
\centering
\subfloat{
\begin{tikzpicture}[node distance = 0cm, scale=0.8, transform shape]
\begin{axis}[xtick=data, xticklabels={1, 2, 4, 8, 16}, xlabel=Number of nodes,ylabel= Speed-up,legend pos= north west, legend style={font=\tiny}]
    \addplot+[color=red,mark=otimes, mark size=2,error bars/.cd, y dir=both,y explicit] coordinates {
	(0,1)
	(1,1.336956522)
	(2,1.618421053)
	(3,1.597402597)
	(4,1.481927711)
    };
    \addplot+[color=blue,mark=x, mark size=2,error bars/.cd, y dir=both,y explicit] coordinates {
	(0,1)		
	(1,1.829959514)
	(2,2.226600985)
	(3,3.503875969)
	(4,4.072072072)
    };
    \addplot+[color=green,mark=diamond, mark size=2, error bars/.cd, y dir=both,y explicit] coordinates {
	(0,1)
	(1,1.969804618)			
	(2,2.691747573)
	(3,4.800865801)
	(4,6.562130178)
	};
    \addplot+[color=black,mark=square, mark size=2, error bars/.cd, y dir=both,y explicit] coordinates {
	(0,1)
	(1,1.520356234)			
	(2,3.002512563)
	(3,4.178321678)
	(4,6.947674419)
    };
    \legend{97M, 291M, 582M, 1163M}
    \end{axis}
\end{tikzpicture}
\label{fig:speed-up}
}
\subfloat{
\begin{tikzpicture}[node distance = 0cm, scale=0.8, transform shape]
\begin{axis}[xtick=data, xticklabels={1, 2, 4, 8, 16}, xlabel=Number of nodes,ylabel=Scaled speed-up,legend pos= north east, legend style={font=\tiny}]    
	\addplot+[color=red,mark=otimes, mark size=2,error bars/.cd, y dir=both,y explicit] coordinates {
	(0,1)
	(1,0.668478261)
	(2,0.404605263)
	(3,0.199675325)
	(4,0.092620482)
    };
    \addplot+[color=blue,mark=x, mark size=2,error bars/.cd, y dir=both,y explicit] coordinates {
	(0,1)		
	(1,0.914979757)
	(2,0.556650246)
	(3,0.437984496)
	(4,0.254504505)
    };    
    \addplot+[color=green,mark=diamond, mark size=2, error bars/.cd, y dir=both,y explicit] coordinates {
	(0,1)
	(1,0.984902309)			
	(2,0.672936893)
	(3,0.600108225)
	(4,0.410133136)
	};
    \addplot+[color=black,mark=square, mark size=2, error bars/.cd, y dir=both,y explicit] coordinates {
	(0,1)
	(1,0.760178117)			
	(2,0.750628141)
	(3,0.52229021)
	(4,0.434229651)
    };
    \legend{97M, 291M, 582M, 1163M}
    \end{axis}
\end{tikzpicture}
\label{fig:scaled_speed-up}
}
\caption{Speed-up and scaled speed-up as a function of numbers of nodes, for various numbers of facts.}
\label{fig:speed-ups}
\end{figure}
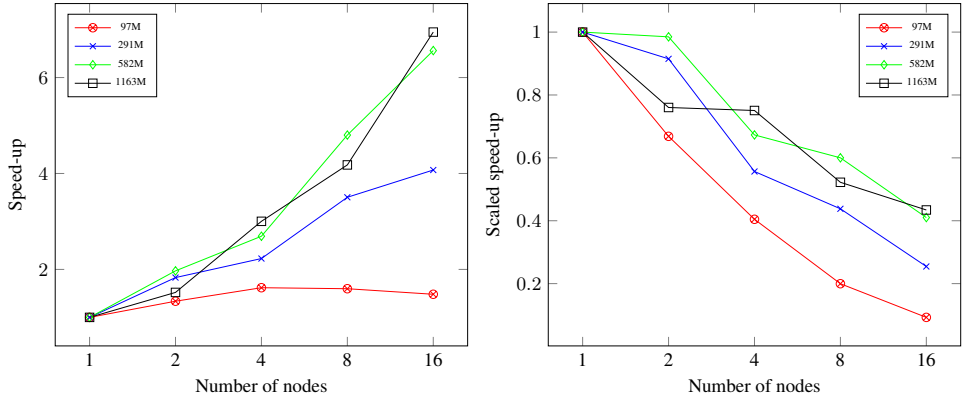

\subsection{Platform}\label{sec:Platform}
Our evaluation platform is the Kay supercomputer located at Irish
Centre for High-End Computing (ICHEC). We choose up to 17 nodes from
the system, and each node we have used contains a 20-core Intel Xeon
Gold 6148 (Skylake) processor running at 2.4GHz with 192GB of RAM and a
single 400GB SSD local disk. The operating system is Linux kernel
version 3.10.0-693 and the software stack consists of Spark version
2.3.1, Hadoop version 2.7.3, Scala version 2.11.8 and Java version
1.8.0\_191.

%{\bf System configuration.}
For Spark, we set the following system parameters:
\textit{spark\_worker\_memory} and \linebreak \textit{spark\_executor\_memory}
are set to 160GB and \textit{spark\_worker\_cores} is to 20. In all
our experiments, the operations of input file reading are on the HDFS
system using the SSD on each node.  We measure runtime as the elapsed
time from job submission to the job being reported as finished and we
record the mean value based on three measurements. 

\subsection{Results}\label{sec:Results}

Table~\ref{tab:number_of_conclusions} provides insight in terms of reasoning.
Specifically, it is evident that queries for \emph{DEMO}, \emph{DRUG} and 
\emph{REAC} are generating comparable numbers of conclusions, \emph{OUTC}
generates approximately half compared to the aforementioned queries, while \emph{RPSR}
generates only a fraction of conclusions. The number of conclusions for each
set of queries is a function of the number of rows in the corresponding file
(see Table~\ref{tab:Number_of_rows_per_file}) and the number of executed queries 
(see Table~\ref{tab:number_of_SQL_queries}). However, providing the exact function that
would allow an accurate prediction of the number of conclusions, based on the number of rows and
executed queries, is out of the scope of this work.

%%%%%%%%%%%%%%%%%%%%%%%%%%%%%%%%%%%%%%%%%%%%%%%%%%%%%%%%%%%%%%%%%%%%%%%%%%%%%%%%%%%%%%%%%
%%% Time in seconds as a function of queries, for various numbers of facts and nodes. %%%
%%%%%%%%%%%%%%%%%%%%%%%%%%%%%%%%%%%%%%%%%%%%%%%%%%%%%%%%%%%%%%%%%%%%%%%%%%%%%%%%%%%%%%%%%
\begin{figure}[t]
%\begin{subfigure}{.498\linewidth}
\centering
\subfloat{
\begin{tikzpicture}[node distance = 0cm, scale=0.8, transform shape]
\begin{axis}[xtick=data, xticklabels={DEMO, DRUG, OUTC, REAC, RPSR}, xlabel=Queries for 97M facts,ylabel= CPU time (s),legend pos= north east, legend style={font=\tiny}]
    \addplot+[color=red,mark=otimes, mark size=2,error bars/.cd, y dir=both,y explicit] coordinates {
	(0,95)
	(1,77)
	(2,33)
	(3,51)
	(4,21)
    };
    \addplot+[color=blue,mark=x, mark size=2,error bars/.cd, y dir=both,y explicit] coordinates {
	(0,74)		
	(1,53)
	(2,29)
	(3,41)
	(4,17)
    };
    \addplot+[color=green,mark=diamond, mark size=2, error bars/.cd, y dir=both,y explicit] coordinates {
	(0,54)
	(1,39)			
	(2,22)
	(3,27)
	(4,11)
	};
    \addplot+[color=black,mark=square, mark size=2, error bars/.cd, y dir=both,y explicit] coordinates {
	(0,57)
	(1,34)			
	(2,19)
	(3,25)
	(4,9)
    };
    \addplot+[color=magenta,mark=triangle, mark size=2, error bars/.cd, y dir=both,y explicit] coordinates {
	(0,59)
	(1,30)			
	(2,14)
	(3,19)
	(4,11)	
    };
    \legend{1 node, 2 nodes, 4 nodes, 8 nodes, 16 nodes}
    \end{axis}
\end{tikzpicture}
\label{fig:eval-97M-sub}
}
\subfloat{
\begin{tikzpicture}[node distance = 0cm, scale=0.8, transform shape]
\begin{axis}[xtick=data, xticklabels={DEMO, DRUG, OUTC, REAC, RPSR}, xlabel=Queries for 291M facts,ylabel= CPU time (s),legend pos= north east, legend style={font=\tiny}]
    \addplot+[color=red,mark=otimes, mark size=2,error bars/.cd, y dir=both,y explicit] coordinates {
	(0,391)
	(1,301)
	(2,90)
	(3,203)
	(4,22)
    };
    \addplot+[color=blue,mark=x, mark size=2,error bars/.cd, y dir=both,y explicit] coordinates {
	(0,182)	
	(1,168)
	(2,65)
	(3,119)
	(4,18)
    };
    \addplot+[color=green,mark=diamond, mark size=2, error bars/.cd, y dir=both,y explicit] coordinates {
	(0,125)
	(1,109)			
	(2,55)
	(3,98)
	(4,13)
	};
    \addplot+[color=black,mark=square, mark size=2, error bars/.cd, y dir=both,y explicit] coordinates {
	(0,100)
	(1,82)			
	(2,35)
	(3,87)
	(4,11)
    };
    \addplot+[color=magenta,mark=triangle, mark size=2, error bars/.cd, y dir=both,y explicit] coordinates {
	(0,88)
	(1,47)			
	(2,27)
	(3,37)
	(4,9)	
    };
    \legend{1 node, 2 nodes, 4 nodes, 8 nodes, 16 nodes}
    \end{axis}
\end{tikzpicture}
\label{fig:eval-291M-sub}
}\\
\subfloat{
\begin{tikzpicture}[node distance = 0cm, scale=0.8, transform shape]
\begin{axis}[xtick=data, xticklabels={DEMO, DRUG, OUTC, REAC, RPSR}, xlabel=Queries for 582M facts,ylabel= CPU time (s),legend pos= north east, legend style={font=\tiny}]
    \addplot+[color=red,mark=otimes, mark size=2,error bars/.cd, y dir=both,y explicit] coordinates {
	(0,827)
	(1,773)
	(2,231)
	(3,521)
	(4,26)
    };
    \addplot+[color=blue,mark=x, mark size=2,error bars/.cd, y dir=both,y explicit] coordinates {
	(0,443)
	(1,392)
	(2,124)
	(3,297)
	(4,20)
    };
    \addplot+[color=green,mark=diamond, mark size=2, error bars/.cd, y dir=both,y explicit] coordinates {
	(0,347)
	(1,303)			
	(2,121)
	(3,241)
	(4,16)
	};
    \addplot+[color=black,mark=square, mark size=2, error bars/.cd, y dir=both,y explicit] coordinates {
	(0,205)
	(1,163)			
	(2,67)
	(3,127)
	(4,12)
    };
    \addplot+[color=magenta,mark=triangle, mark size=2, error bars/.cd, y dir=both,y explicit] coordinates {
	(0,146)
	(1,87)			
	(2,40)
	(3,76)
	(4,10)	
    };
    \legend{1 node, 2 nodes, 4 nodes, 8 nodes, 16 nodes}
    \end{axis}
\end{tikzpicture}
\label{fig:eval-582M-sub}
}
\subfloat{
\begin{tikzpicture}[node distance = 0cm, scale=0.8, transform shape]
\begin{axis}[xtick=data, xticklabels={DEMO, DRUG, OUTC, REAC, RPSR}, xlabel=Queries for 1163M facts,ylabel= CPU time (s),legend pos= north east, legend style={font=\tiny}]
    \addplot+[color=red,mark=otimes, mark size=2,error bars/.cd, y dir=both,y explicit] coordinates {
	(0,1809)
	(1,1552)
	(2,587)
	(3,1351)
	(4,34)
    };
    \addplot+[color=blue,mark=x, mark size=2,error bars/.cd, y dir=both,y explicit] coordinates {
	(0,954)	
	(1,835)
	(2,333)
	(3,814)
	(4,28)
    };
    \addplot+[color=green,mark=diamond, mark size=2, error bars/.cd, y dir=both,y explicit] coordinates {
	(0,576)
	(1,523)			
	(2,201)
	(3,448)
	(4,18)
	};
    \addplot+[color=black,mark=square, mark size=2, error bars/.cd, y dir=both,y explicit] coordinates {
	(0,492)
	(1,431)			
	(2,160)
	(3,371)
	(4,18)
    };
    \addplot+[color=magenta,mark=triangle, mark size=2, error bars/.cd, y dir=both,y explicit] coordinates {
	(0,300)
	(1,216)			
	(2,88)
	(3,173)
	(4,15)	
    };
    \legend{1 node, 2 nodes, 4 nodes, 8 nodes, 16 nodes}
    \end{axis}
\end{tikzpicture}
\label{fig:eval-1163M-sub}
}
\caption{Time in seconds as a function of queries, for various numbers of facts and nodes.}
\label{fig:eval-queries}
\end{figure}
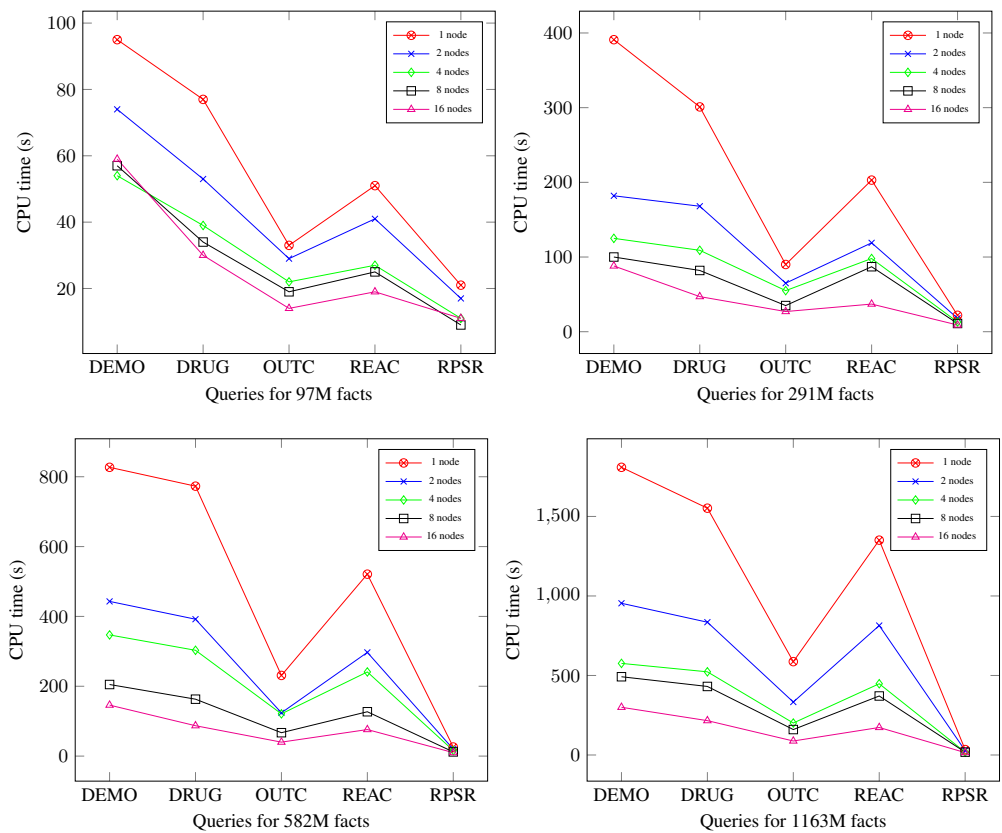

Figure~\ref{fig:eval-whole} shows the scalability results of the implementation
for increasing number of facts. The implementation follows a fairly linear scalability
up to 1.16 billion facts when the number of nodes ranges from 1 to 16. From a 
practical point of view, the initial dataset (four calendar years) can be processed
with 16 nodes in 83 seconds, while 12 copies (corresponding to almost half a century)
can be processed in 5 minutes and 44 seconds. For comparison, the initial dataset
would require more than 5 days with RuleRS~\cite{DBLP:journals/ail/IslamG18}, while 
auditing the 12 copies with RuleRS would require approximately 2 months.
Even though the proposed approach in this work cannot be directly compared to RuleRS, given the fact that RuleRS is based on a serial implementation, our results show a significant scalability advantage of the proposed inference rules.

Figure~\ref{fig:speed-ups} depicts speed-ups and scaled speed-ups\footnote{Speed-up is
calculated as: $\frac{{runtime}_{1node}}{{runtime}_{Nnodes}}$, while scaled speed-up is
calculated as: $\frac{{runtime}_{1node}}{N~*~{runtime}_{Nnodes}}$, where ${runtime}_{1node}$ 
is the required run time for one node, $N$ is the number of nodes and ${runtime}_{Nnodes}$ is 
the required run time for N nodes.} 
for increasing number
of nodes, for various number of facts. It is evident that 97M of facts is relatively small 
input in order to show the benefits of parallelization, this is attributed to the fact
that the majority of time is dedicated to reading the input. On the other hand, larger 
inputs highlight the advantages of the distributed implementation. However, the speed-ups
are sub-linear regardless of the number of facts or nodes. 
% This is a clear indication that 
% adding more nodes would be cost effective only if the implementation is improved. 
Nonetheless, the results are encouraging in terms of a proof of concept.

Figure~\ref{fig:eval-queries} presents the required time in order to execute each set
of queries separately (including reasoning over generated facts). The required time in
declining order is as follows: \emph{DEMO}, \emph{DRUG}, \emph{REAC}, \emph{OUTC} and 
\emph{RPSR}. This is consistent with Tables~\ref{tab:Number_of_rows_per_file} 
and~\ref{tab:number_of_conclusions} since larger files require more time to be read, while
more conclusions mean both longer reasoning time and more generated facts from the 
executed queries. Finally, once the input is large enough, there is a clear trend where adding
more nodes leads to faster runtimes.

\section{Conclusion and Future Work}\label{sec:CONCLUSION}
In this paper, we introduced a scalable defeasible logic that allows reasoning over large amounts of data. In particular, we proposed new inference rules for defeasible reasoning, discussed the theoretical properties of the new defeasible logic and ran experiments over an FDA case study (with rules encoding FDA regulations over publicly
available FDA datasets). Our experimental results indicate that this method can be applied to billions of facts.

In future work, we plan to develop a generic implementation of a parallel reasoner over the logic we propose
in this work. In addition, we plan to study how the proposed inference rules can be extended further in order to model more complex constructs while retaining scalability. In particular, a potential direction could be the introduction of a scalable Defeasible Deontic Logic as an alternative to the one presented in~\cite{DBLP:journals/jphil/GovernatoriORS13}. 
Another direction could be the extension of the proposed defeasible logic in this work to the BOID (Belief, Obligation, Intention, Desire) architecture~\cite{DBLP:journals/aamas/GovernatoriR08}. Such approaches would facilitate reasoning in the legal context, thus providing scalable solutions for processing large amounts of legal documents.

\section*{Acknowledgments}
We thank the referees for their comments, which helped improve this paper.
%This work was partially supported by the PlanetData NoE (FP7:ICT-2009.3.4, \#257641).

\bibliography{Scalable_Defeasible_Logic}

\appendix

\include{appendix}
\end{document}

%% file: appendix.tex
\newcounter{temp}

\section{Relative Inference Strength}  \label{app:ris}

%\setcounter{temp}{\value{theorem}}
%\setcounter{theorem}{\ref{prop:contain}}
%\addtocounter{theorem}{-1}

\begin{proposition}  \label{prop:contain}
\ignore{
$\Delta \subset \pl \subset \lambda$
}
$\Delta \subset \pl^* \subset \pl \subset \lambda$
\end{proposition}
\skipit{
\begin{proof}
\ignore{
The first containment follows immediately from (1) of the $\pl$ inference rule.
The $\lambda$ inference rule is essentially the $\pl$ inference rule with condition (2.3) omitted.
The second containment then follows.
}
The first containment follows immediately from (1) of the $\pl^*$ inference rule.
The only difference between $\pl^*$ and $\pl$ is in (2.3.2),
and the clause for $\pl^*$ implies the clause for $\pl$.
The second containment follows.
The $\lambda$ inference rule is essentially the $\pl$ inference rule with condition (2.3) omitted.
The third containment then follows.

% Strictness .... need to change statement of prop
Strictness is shown with straightforward examples.
Strictness of the first containment is shown by $D$ consisting only of $\Rightarrow p$.
Strictness of the second containment is shown by the standard example distinguishing team and individual defeat:
$D$ consists  of:
\[
\begin{array}{lrcl}
r_1:        &             & \Rightarrow & \phantom{\neg} p \\
r_2:        &             & \Rightarrow & \phantom{\neg} p \\
r_3:        &             & \Rightarrow &  \neg p \\
r_4:        &             & \Rightarrow &  \neg p \\
\end{array}
\]
with $r_1 > r_3$ and $r_2 > r_4$.
Then we can conclude $+\pl p$ but not $+\pl^* p$.

Example \ref{ex:lambda-consistent} shows the strictness of the third containment
since $+\lambda q$ is proved but $+\pl q$ cannot be proved.
\end{proof}
}

It is straightforward to see that $\lambda$ is not consistent.

\begin{example}      \label{ex:lambda-consistent}
Consider the defeasible theory
\[
\begin{array}{lrcl}
r: &           & \Rightarrow & \phantom{\neg} q \\
s: &           & \Rightarrow & \neg q \\
\end{array}
\]
with empty superiority relation.

Then we can infer $+\lambda q$ and $+\lambda \neg q$,
but cannot infer $\PD{q}$ nor $\PD{\neg q}$.
Thus $\lambda$ is not consistent.
Furthermore, we cannot infer
$+\pl q$ nor $+\pl \neg q$.
\end{example}

%\setcounter{temp}{\value{theorem}}
%\setcounter{theorem}{\ref{prop:consistent}}
%\addtocounter{theorem}{-1}

\begin{proposition}   % \label{prop:consistent}
The inference rule $+\pl$ is consistent.
\end{proposition}
\skipit{
\begin{proof}
Suppose, for some defeasible theory $D$, and some proposition $q$,
that $+\pl q$ and $+\pl \neg q$ are consequences of $D$.

If $\PD{\neg q} \in P_\Delta$ but $\PD{q} \notin P_\Delta$ then,
when attempting to prove $\PD{q}$,
neither (1) nor (2.2) of the $\pl$ inference rule hold
and, thus, $\PD{q}$ cannot be proved.
This contradicts our original supposition, so this case cannot occur.
Similarly, the case where $\PD{q} \in P_\Delta$ but $\PD{\neg q} \notin P_\Delta$ cannot occur.

In the third case, neither $\PD{q}$ nor $\PD{\neg q}$ are consequences.
Since $+\pl q$ is a consequence and (1) does not hold,
(2.1) of the $\pl$ inference rule must hold for some rule $r$ for $q$.
Symmetrically, there is a rule $s$ for $\neg q$ such that (2.1) holds.
Consequently, by Proposition \ref{prop:contain}, for each $\alpha \in A(s)$, $+\lambda \alpha \in P_\lambda$.
Hence, to infer $+\pl q$, there must be a rule $t$ for $q$ with $t > s$ and
for each $\beta \in A(t)$, $+\pl \beta$ is provable and thus $+\lambda \beta \in P_\lambda$.
But then, to infer $+\pl \neg q$, there must be a rule $t'$ with $t' > t$ and
for each $\gamma \in A(t)$, $+\pl \gamma$ is provable.
And so on.
This creates an infinite chain of rules, each superior to the previous rule.
No rule can be repeated, since $>$ is acyclic.
However, the chain cannot be infinite, since the set of rules is finite.
This contradiction shows that this case cannot occur.
 
Thus, by exclusion, both $\PD{q}$ and $\PD{\neg q}$ are consequences,
and the result is proved.
\end{proof}
}

It follows immediately from Propositions~\ref{prop:consistent} and \ref{prop:contain}
that $DL(\pl^*)$ also is consistent.

\begin{corollary}
The inference rule $\pl^*$ is consistent.
\end{corollary}

The next two examples show that $\partial $ and $\pl$ are incomparable in inference strength.

% $\pl \not\subseteq \partial$
\begin{example}    \label{ex:plpartial}
Consider the defeasible theory
\[
\begin{array}{lrcl}
r: &           & \Rightarrow & \phantom{\neg} q \\
s: & \neg q & \rightarrow & \neg q \\
\end{array}
\]
with $r > s$.

Then $\PD{\neg q}$ cannot be inferred, and so $\ppd{q}$ is inferred.
On the other hand, $\MD{\neg q}$ also cannot be inferred, and so $\pd{q}$
cannot be inferred.
Consequently, $\pl \not\subseteq \partial$.

This comes about because of the different treatments of opposing strict inferences
in the two inference rules.
\end{example}

% $+\partial  \not\subseteq \pl$
\begin{example}      \label{ex:partialpl}
Consider the defeasible theory
\[
\begin{array}{lrcl}
r: &           & \Rightarrow & \phantom{\neg} q \\
s: &           & \Rightarrow & \neg q \\
t: &           & \Rightarrow & \phantom{\neg} p \\
u: &     q    & \Rightarrow & \neg p \\
\end{array}
\]
with no superiority relation.

Then we can infer $-\partial q$ and $+\lambda q$.
Consequently, we can infer $+\partial p$, but not $+\pl p$.
Hence, $\partial \not\subseteq \pl$.

This comes about because the inference rules for $+\pl$ and $+\partial$ differ at (2.3.1):
$+\pl$ requires $+\lambda \alpha \notin P_\lambda$ while $+\partial$ requires $-\partial \alpha \in P(1..i)$.
\end{example}

\begin{proposition}   \label{prop:partiallambda}
$\partial \subset \lambda$
and
$\partial^* \subset \lambda$
\end{proposition}
\skipit{
\begin{proof}
The $\lambda$ inference rule has no condition (2.3),
and replaces the condition for $\MD{\non q} \in P$ for $\partial$ with $\PD{\non q} \notin P$.
By the coherence of defeasible logics \cite{TOCL10},
$\PD{\non q} \notin P$ is a weaker condition.
Hence every inference that $+\partial$ can make can be duplicated by $+\lambda$.
The result then follows.

The same argument applies to show that $\partial^* \subseteq \lambda$.

Strictness in both cases is straightforward,
using the same example and argument 
(Example~\ref{ex:lambda-consistent}) as in the proof of Proposition \ref{prop:contain}
for the strictness of the third containment.
\end{proof}
}

\ignore{ now obsolete
\begin{proposition} %  \label{prop:partiallambda}
$\pl \not\subseteq \partial$ and $\pl \not\subseteq \partial^*$
$\pl^* \not\subseteq \partial$ and $\pl^* \not\subseteq \partial^*$

\end{proposition}
\skipit{
\begin{proof}
Consider the theory $D$ from the proof of Theorem \ref{thm:pl_v_dl}, consisting of 
\[
\begin{array}{lrcl}
        &             & \Rightarrow & \phantom{\neg} p \\
        & \neg p  & \rightarrow &  \neg p \\
\end{array}
\]
Then $+\lambda p$, $+\pl p$, and $+\pl^* p$ are consequences of $D$, as is $-\Delta p$.
On the other hand,  $-\Delta \neg p $ is not a consequence and hence $+\partial p$ is not a consequence.
Thus $\pl \not\subseteq \partial$.
By the same argument, $\pl \not\subseteq \partial^*$.
\end{proof}
}
}

We establish the lack of any additional containments in
Figure~\ref{fig:inclusion2}
using the following three examples.

\begin{example}    \label{ex:lambdasigma}
Consider the defeasible theory
\[
\begin{array}{lrcl}
r: &           & \Rightarrow & \phantom{\neg} q \\
s: & \neg q & \rightarrow & \neg q \\
\end{array}
\]
with $s > r$.

Then we can infer $+\lambda q$ and $+\pl^* q$, but not $+\sigma_{\delta^*} q$.
This arises because the inference rule for $+\lambda$ and $+\pl^*$ requires only that $\PD{\non q}$ is not inferred,
while the inference rule for $+\sigma_{\delta^*}$ must establish $-\delta^* \non q$.
In this case, $-\delta^* \non q$ cannot be inferred.
Thus $\pl^* \not\subseteq \sigma_{\delta^*}$.

It then follows that
$\pl^* \not\subseteq X$, for any $X$ considered in \cite{GM17}
(since $X \subseteq \sigma_{\delta^*}$ \cite{TOCL10}),
and also $\pl \not\subseteq X$ and $\lambda \not\subseteq X$
(since $\pl^* \subseteq \pl$).
\end{example}

We use $\sigma_X$ to denote any of the support inference rules
$\sigma_{\delta^*}$, $\sigma_{\delta}$, $\sigma_{\partial^*}$, and $\sigma_\partial$.

% $\sigma_X \not\subseteq \lambda$ 

\begin{example}      \label{ex:sigmalambda}
Consider the defeasible theory
\[
\begin{array}{lrcl}
r: &           & \Rightarrow & \phantom{\neg} q \\
s: &           & \rightarrow & \neg q \\
\end{array}
\]
with no superiority relation.

Then we can infer $+\sigma_X q$ but not $+\lambda q$.
Thus $\sigma_X \not\subseteq \lambda$.
This comes about because the inference rules for $+\sigma_X$
ignores the possibility of strict inference of $\non q$, 
while the inference rule for $+\lambda$ does not.

Hence $\sigma_X \not\subseteq \lambda$, for any $X$.
\end{example}

\begin{example}      \label{ex:delta*pl}
Consider the defeasible theory
\[
\begin{array}{lrcl}
r: &           & \Rightarrow & \phantom{\neg} p \\
s: &           & \Rightarrow & \neg p \\
   &           & \Rightarrow & \phantom{\neg} q \\
   &     p    & \Rightarrow & \neg q \\
\end{array}
\]
with $s > r$.

Then we can infer $+\lambda p$, and hence cannot infer $+\pl q$.
On the other hand, we can infer $-\sigma_{\delta^*} p$, since $s > r$,
and hence we can infer $+\delta^* q$.
Thus $\delta^* \not\subseteq \pl$.

Because $\pl^* \subseteq \pl$ and $\delta^* \subseteq X$ for every $X$ discussed in \cite{GM17} except $\Delta$,
we can conclude that neither $\pl$ nor $\pl^*$ contains any $X$ discussed in \cite{GM17} except $\Delta$.
\end{example}

%\setcounter{temp}{\value{theorem}}
%\setcounter{theorem}{\ref{thm:fig2}}
%\addtocounter{theorem}{-1}

\begin{theorem}   % \label{thm:fig2}
The containments illustrated in Figure \ref{fig:inclusion2} hold and are strict.
In addition, $\partial^* \subset \lambda$ holds.
There are no other missing containments in the figure.
\end{theorem}
\skipit{
\begin{proof}
The containments on the top row of the diagram are established in Proposition \ref{prop:contain}.
The containments on and between the lower two rows are established in \cite{TOCL10,GM17},
including their strictness and the lack of any other containments among them.
The containments $\partial^* \subset \lambda$ and $\partial \subset \lambda$
are established in Proposition \ref{prop:partiallambda}.

Example \ref{ex:lambdasigma} shows that no tag in the lower rows contains a tag in the upper row.
Furthermore, Example \ref{ex:delta*pl} shows that $\pl^*$ and $\pl$ do not contain any tag on the lower rows, except for $\Delta$.
and Example \ref{ex:sigmalambda} shows that $\lambda$ does not contain any of the $\sigma_X$ tags.
Examples showing that containments are strict are straightforward and left to the reader.
\end{proof}
}

\section{Complexity}  \label{app:complexity}

In this appendix we prove results on the complexity of $\DL(\pl)$.

As a result of the structure of the inference rules
it is straightforward to compute the consequences of $\Delta$ and $\lambda$ efficiently.
\begin{lemma}  %  \label{lemma:linearDL}
The $\Delta$ and $\lambda$ closures, $P_\Delta$ and $P_\lambda$, of a propositional defeasible theory can be computed in linear time.
\end{lemma}
\skipit{
\begin{proof}
(Sketch)
The inference rule for $+\Delta$ is already treated in \cite{linear}.
Alternatively, this inference rule is essentially treating strict rules as definite clauses,
where negative literals ($\neg p$) are considered as atoms (e.g. $not\_p$).
Such inference can be done in time linear in the size of facts and strict rules
\cite{DowlingGallier}.

Similarly, the inference rule for $+\lambda$ essentially treats strict and defeasible rules as definite clauses,
with an extra condition about $\Delta$ consequences.
Once the $\Delta$ consequences have been computed,
it takes constant extra time for each rule to check the extra condition.
Consequently, the inference of $+\lambda$ consequences takes time linear in the size
of facts, strict rules and defeasible rules.

Similarly, $-\Delta$ and $-\lambda$ consequences (see \ref{app:relexp})
are also computed in linear time
(although this information is not necessary for the results in this appendix).
%Note that this proof holds for both normal and well-founded propositional defeasible logics.
\end{proof}
}

The inference problem for propositional $\DL(\partial)$ has linear complexity \cite{linear},
and we use the same techniques to show that
the inference problem for propositional $\DL(\pl)$ also has linear complexity.

%\setcounter{temp}{\value{theorem}}
%\setcounter{theorem}{\ref{thm:linear}}
%\addtocounter{theorem}{-1}

\begin{theorem}   %  \label{thm:linear}
The set of all consequences of a propositional defeasible theory can be computed in time linear in the size of the defeasible theory.
Consequently, 
the inference problem for propositional $\DL(\pl)$ can be solved in linear time.
\end{theorem}
\skipit{
\begin{proof}
(Sketch)
We adapt the approach of \cite{linear}.
This is possible largely because the structure of the inference rules for $\partial$ and $\pl$ are the same.
First, observe that the transformations of \cite{TOCL01}  for $\DL(\partial)$
are also correct for $\DL(\pl)$.
These transformation are used in \cite{linear} to reduce the input defeasible theory $D$
to an equivalent theory in simpler form.

There are three transformations in \cite{TOCL01} .
The first, $\mathit{regular}$, separates strict rules from the superiority relation,
and it is straightforward to see that this is valid for a wide range of defeasible logics, including $\DL(\pl)$.
The other two, $\mathit{elim\_dft}$ and $\mathit{elim\_sup}$, 
which are used to eliminate defeaters and the superiority relation respectively,
employ the same technique to achieve their respective aims:
they introduce an intermediate literal in a rule that might be attacked.

For example, roughly speaking, a rule $B \Rightarrow h$ is replaced by
$B \Rightarrow temp$ and $temp \Rightarrow h$, and
a defeater $B' \leadsto \non h$ is replaced by $B' \Rightarrow \non temp$.
Similarly,
if we have rules $r_1 ~:~ B_1 \Rightarrow h$ and $r_2 ~:~ B_1 \Rightarrow \non h$
with $r_1 > r_2$ then these are replaced by
$B_1 \Rightarrow temp_1$, $temp_1  \Rightarrow h$, $B_2 \Rightarrow temp_2$, $temp_2  \Rightarrow \non h$, 
and $temp_1 \Rightarrow \neg temp_2$, where the latter rule encodes $r_1 > r_2$.
In each case, when the defeater or overriding rule is active
the intermediate literal fails to be proved because it is attacked by another rule,
and consequently the application of the original rule is prevented.
Because the structure of the inference rules is the same for $\DL(\pl)$ and $\DL(\partial)$,
the introduction of intermediate literals and the effect of an attacking rule is the same in both logics.
Thus the technique is also correct in $\DL(\pl)$. 

%We adapt the approach of \cite{linear}, but we are able to simplify it greatly for $\DL(\pl)$.
We have already seen that $P_\Delta$ and $P_\lambda$ can be computed in linear time.
Now we can simplify the transformed version of $D$ and deduce some $\pl$ consequences.

Let $C$ be a set of consequences, initially $\emptyset$.
\begin{enumerate}
\item  \label{item:delta}
For each literal $q$:
If $+\Delta q \in P_\Delta$ then delete all defeasible rules for $\non q$, add $+\pl q$ to $C$,
and delete all occurrences of $q$ from the body of rules.
\item  \label{item:lambda}
For each literal $q$:
If $+\lambda q \notin P_\lambda$ and $q$ occurs in the body of a rule, delete the rule.
\item \label{item:strict}
Delete all strict rules.
\ignore{
\item \label{item:>}
For each rule $r$, with head $q$ say,
If there is a rule $s$ for $\non q$ with empty body and $s > r$,
and there is no rule $t$ for $\non q$ with $s > t$,
then delete $r$.
}
\end{enumerate}

Simplification \ref{item:delta} is justified by (2.2) and (1) of the inference rule.
Simplification \ref{item:lambda} is justified by (2.3.1) of the inference rule,
and by Proposition \ref{prop:contain} (which implies that such rules cannot be used in (2.1) for $\non q$).
Simplification \ref{item:strict} is justified because all definite consequences are already available in $P_\Delta$
and, as a result of the $\mathit{regular}$ transformation, no other use is made of these rules.
%Simplification \ref{item:>} comes from \cite{Lam.2011}

The simplified theory $D'$ incorporates the all the effects of references to $P_\Delta$ and $P_\lambda$.
Consequently, the transition system of \cite{linear} applies also to $D'$ with initial consequences $C$, for $\DL(\pl)$.
In fact, only the transitions numbered 2, 4, 5, and 8 are needed, since the remaining transitions involve strict rules or negative tags, though 5 is modified by dropping the reference to $-\Delta\non q$.
The simplifications above can be viewed as variants of transitions:
simplification \ref{item:delta} corresponds to transitions 6 and 1;
and
simplification \ref{item:lambda} corresponds to transition 10.
Simplification \ref{item:strict} is essentially redundancy removal, given $C$.
Furthermore, the data structure used in \cite{linear} to achieve linear complexity
in application of the transition system is also applicable to $\DL(\pl)$.

Thus all positive consequences of a propositional defeasible theory $D$ in $\DL(\pl)$
can be computed in time linear in the size of $D$.
\end{proof}
}

%\setcounter{theorem}{\value{temp}}
%\setcounter{temp}{\value{theorem}}
%\setcounter{theorem}{\ref{cor:exp}}
%\addtocounter{theorem}{-1}

\begin{corollary}  % \label{cor:exp}
The set of all consequences of a defeasible theory can be computed in time exponential in the size of the defeasible theory.
Furthermore, the inference problem for defeasible theories is EXPTIME-complete.
\end{corollary}
\skipit{
\begin{proof}
Construct a propositional defeasible theory $D'$ from the original defeasible theory $D$ by taking all variable-free instances of all rules using the constants that appear in $D$.
Two instances of rules are related by the superiority relation iff the rules of which they are instances are so related.
Let $n$ be the maximum number of variables in a rule of $D$ and $c$ be the number of constants in $D$.
Then there are at most $c^n$ propositional instances of a rule of $D$, and at most $c^{2n}$ derived superiority statements for each superiority statement in $D$.
Since both $n$ and $c$ may be O($|D|$),
the size of $D'$ is O($|D|^{2|D|}$),
which is O($2^{p(|D|)}$), for a polynomial $p$.

$D$ and $D'$ have the same consequences.
By Theorem \ref{thm:linear} the consequences of $D'$ can by computed in linear time in the size of $D'$,
which is EXPTIME in the size of $D$.

The inference problem is shown EXPTIME-complete
by reduction of the same problem for Datalog (see \cite{DEGV}, Theorem 4.5).
% NB minor adjustment of DEGV is needed to make it range-restricted.
Each Datalog rule is expressed as a defeasible rule.
A positive literal is inferred in
$\DL(\pl)$ iff it is inferred in Datalog.
\end{proof}
}

%\setcounter{theorem}{\value{temp}}
%\setcounter{temp}{\value{theorem}}
%\setcounter{theorem}{\ref{thm:Pcomplete}}
%\addtocounter{theorem}{-1}

\begin{theorem}  % \label{thm:Pcomplete}
The inference problem for propositional defeasible logics is P-complete.
\end{theorem}
\skipit{
\begin{proof}
We show that the inference problem for $\PD{}$ is P-complete,
by reduction of the Horn satisfiability problem, which is P-complete \cite{CookNguyen}.  
For completeness, we first specify this problem.
A Horn clause is a disjunction of literals containing at most one positive literal.

\ \\
\noindent
\textbf{The Horn Satisfiability Problem}

\noindent
\textbf{Instance} \\
A set $H$ of propositional Horn clauses.

\noindent
\textbf{Question} \\
Is $H$ satisfiable, that is, is there an assignment of Boolean values to propositional variables
such that each clause of $H$ evaluates to true?
\ \\

In the reduction,
each of the propositional variables in the Horn satisfiability problem is represented by itself,
and we add an extra propositional variable $\false$.
For clarity, we write the Horn clauses in the logic programming style.

For every Horn clause of the form
\[
\begin{array}{rcl}
A & \leftarrow & B_1, \ldots, B_n
\end{array}
\]
the defeasible theory contains the strict rule
\[
\begin{array}{rcl}
B_1, \ldots, B_n & \rightarrow & A
\end{array}
\]

Similarly, for every Horn clause of the form
\[
\begin{array}{rcl}
 & \leftarrow & B_1, \ldots, B_n
\end{array}
\]
the defeasible theory contains the strict rule
\[
\begin{array}{rcl}
B_1, \ldots, B_n & \rightarrow & \false
\end{array}
\]

It is straightforward to show that 
$\PD{q}$ is inferred by a defeasible logic iff
$q$ is true in every model of the definite clause subset of $H$, and
$\PD{\false}$ is inferred by a defeasible logic iff
$H$ is unsatisfiable.

Strict inference is a part of any defeasible logic, so the result applies to all defeasible logics.
Even without a separate notion of strict inference,
the proof extends easily to any inference rule that allows the chaining of defeasible or strict rules,
since the superiority relation and conflicting rules do not arise in the reduction.
This includes all defeasible logics we are aware of.
\end{proof}
}

\section{Relative Expressiveness}   \label{app:relexp}

\newcommand{\p}{{\em ~~Poss~~}}
\newcommand{\np}[1]{{\em ~~NP(#1)~~}}

In this appendix we prove Theorems \ref{thm:pl_v_dl} and  \ref{thm:dl_v_pl}.

Relative expressiveness involves both positive and negative tags, so we first introduce the inference rules for $-\lambda$ and $-\pl$.
These inference rules are a kind of negation of the corresponding positive inference,
under the Principle of Strong Negation \cite{flexf}.
However, the notion of strong negation must be extended to address expressions of the form
$t \: \alpha \notin P$, which were not considered in  \cite{flexf}.
In these cases we define the strong negation of $t \: \alpha \notin P$ to be $t \: \alpha \in P$.

The closure $P_\Delta$ must be closed under both $+\Delta$ and $-\Delta$ inference rules,
that is, it must contain all $+\Delta$ and $-\Delta$ consequences.

The $-\lambda$ inference rule is as follows.

\begin{tabbing}
$-\lambda$: \=We may append $P(\imath + 1) = -\lambda q$ if both \\
\> (1) \=$-\Delta q \in P_{\Delta}$ and \\
\> (2)	\>(2.1) $\forall r \in R_{sd}[q] ~ \exists \alpha \in A(r): -\lambda \alpha \in P(1..\imath)$ or \\
\> \>(2.2) $+\Delta \non q \in P_{\Delta}$ \\
\end{tabbing}

The $\lambda$ closure $P_{\lambda}$
contains all $+\lambda$ and $-\lambda$ consequences of $D$.

\begin{tabbing}
$-\pl$: \=We may append $P(\imath + 1) = -\pl q$ if both \\
\> (1) \=$-\Delta q  \in P_{\Delta}$ and  \\
\> (2)	\>(2.1) $\forall r \in R_{sd}[q] ~ \exists \alpha \in A(r): -\pl \alpha \in P(1..\imath)$ or \\
\> \>(2.2) $+\Delta \non q \in P_{\Delta}$ or \\
\> \>(2.3) \=$\exists s \in R[\non q]$ such that \\
\> \> \>(2.3.1) $\forall \alpha \in A(s): +\lambda \alpha \in P_{\lambda}$ and \\
\> \>\>(2.3.\=2) $\forall t \in R_{sd}[q]$  either \\
\> \>\>\>$\exists \alpha \in A(t): -\pl \alpha \in P(1..\imath)$ or $t \not> s$
\end{tabbing}

To prove the first part of Theorem \ref{thm:pl_v_dl} we employ an analysis
introduced in \cite{table,TOCL01}.
For each proposition $p$ we can identify exactly six
different possible outcomes of the proof theory.
With each outcome we present a simple theory that achieves this outcome.
\begin{itemize}
\item[A:]  $\MD{p} \notin P_\Delta$ and $\ppd{p} \notin P_{\pl}$  \\
$p \rightarrow p$
\item[B:]  $\ppd{p} \in P_{\pl}$ and $\ppd{p} \notin P_\Delta$ and $\MD{p} \notin P_\Delta$  \\
$ \Rightarrow p; p \rightarrow p$
\item[C:]  $\PD{p} \in P_\Delta$ (and also $\ppd{p} \in P_{\pl}$)  \\
$ \rightarrow p$
\item[D:]  $\ppd{p} \in P_{\pl}$ and $\MD{p} \in P_\Delta$  \\
$ \Rightarrow p$
\item[E:]  $\MD{p} \in P_\Delta$ and $\ppd{p} \notin P_{\pl}$  and $\mmd{p} \notin P_{\pl}$\\
$p \Rightarrow p$
\item[F:]  $\mmd{p} \in P_{\pl}$ (and also $\MD{p} \in P_\Delta$)  \\
$\emptyset$, the empty theory
\end{itemize}
Similarly, there are the same six possibilities for $\neg p$.
We can represent the outcomes in terms of a Venn diagram
in Figure \ref{venn}.

\begin{figure}[ht]
\includegraphics[width=10cm,height=6cm]{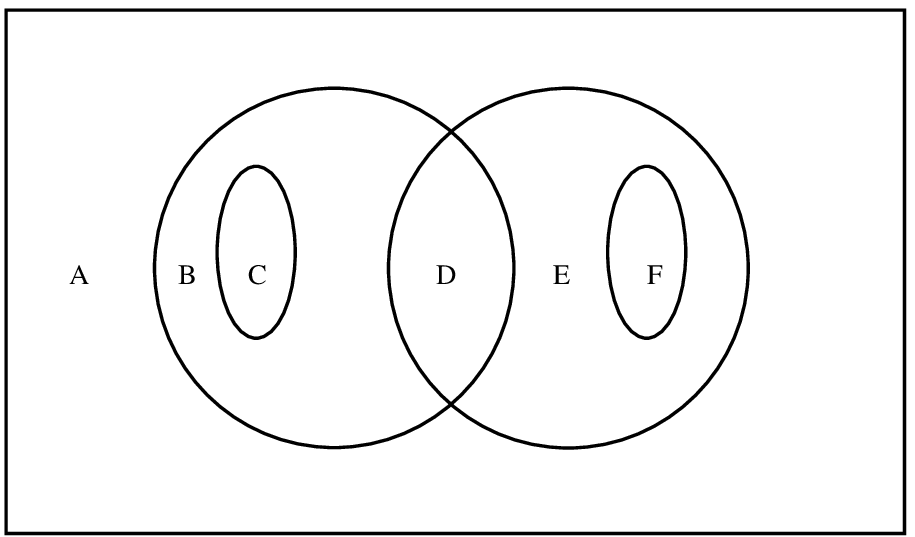}
\caption{Possible outcomes for a single literal in $\DL(\pl)$}
\label{venn}
\end{figure}

In Figure \ref{venn},
the circle on the left -- containing B, C, and D --
represents the literals $p$ such that $+\pl{p}$ can be proved,
and the ellipse inside it (i.e. C)
represents the literals $p$ such that $\PD{p}$ can be proved.
The circle on the right -- containing D, E, and F --
represents the literals $p$ such that $\MD{p}$ can be proved,
and the ellipse inside it (i.e. F)
represents the literals $p$ such that $-\pl{p}$ can be proved.

Due to the relationship between $p$ and $\neg p$, many fewer than
the 36 possible combinations are possible outcomes of the proof theory.

In what follows, $p$ ranges over literals and $\non p$ denotes
the complement of $p$.
We first establish some simple properties that will eliminate many combinations.

\begin{proposition}   \label{prop:props}
Consider a defeasible theory $D$,
with $\Delta$ closure $P_\Delta$ and $\pl$ closure $P_{\pl}$.
\begin{enumerate}
\item
If $\PD{p} \notin P_\Delta$ and $\PD{\non p}  \in P_\Delta$ 
then $+\pl{p} \notin P_{\pl}$

\item
If $\MD{p}  \in P_\Delta$ and $\PD{\non p}  \in P_\Delta$ 
then $-\pl{p} \in P_{\pl}$

\item
% Assume $D$ is finite and $>$ is acyclic.
If $\PD{p} \notin P_\Delta$ and $\PD{\non p} \notin P_\Delta$
then we cannot have both $+\pl{p} \in P_{\pl}$ and $+\pl{\non p} \in P_{\pl}$

\end{enumerate}
\end{proposition}
\begin{proof}
% Statement 1 has been proven in \cite{Billington93}, while
Statements 1 and 2 follow directly from the
proof rules for  $+\pl$ and  $-\pl$.
Statement 3 is proved as follows.

Suppose this combination is possible.
Then, when applying the inference rule $\pl$ for $p$ (and for $\non p$),
(1) does not apply and (2.2) is satisfied.
We also must have $+\lambda p$ and $+\lambda \non p$ because $\pl \subseteq \lambda$
and by definition of D (because we have $+\pl p$ and $+\pl \non p$).
There must be a rule $r$ for $p$ (and one for $\non p$) such that
$\forall \alpha \in A(r) ~ +\pl \alpha \in P_{\pl}$.
Consequently, there is a rule $s$ for $p$ (and one for $\non p$)
such that $\forall \alpha \in A(s) ~ +\lambda \alpha \in P_\lambda$.
Hence (2.3.1) of inference rule $\pl$ does not apply, and so (2.3.2) must.
That is, there is a rule $r$ for $p$ (and one for $\non p$) such that
$\forall \alpha \in A(t) ~ +\pl \alpha \in P_{\pl}$ and $t > s$.
In this way we obtain a chain of rules:
for each rule $s_i$ for $\non p$ there is a superior rule $t_i$ for $p$, and
for each rule $t_i$ for $p$ there is a superior rule $s_i$ for $\non p$.
Since $D$ is finite and $>$ is acyclic, this can never occur.
%\hfill $\Box$ \\
\end{proof}

\ignore{
Property 3 in the above proposition is a weaker form of the \emph{consistency}
of the defeasible logic \cite{TOCL10}.  It demonstrates that any inconsistency
(proving both $p$ and $\neg p$ wrt $\pl$) occurs if there is already inconsistency in
}

In terms of the diagram (Figure \ref{venn}),
the properties of the previous proposition have the following effects:
\begin{enumerate}
\ignore{\item
If $p$ satisfies B, C, or D,
and $\non p$ satisfies B, C, or D,
and $D$ is acyclic
then $p$ satisfies C
and $\non p$ satisfies C (Property 1).
Consequently, if $D$ is acyclic, it is not possible for
$p$ to satisfy B or D,
and $\non p$ to satisfy B, C, or D.
}
\item % 1
If $p$ satisfies A, B, D, E, or F,
and $\non p$ satisfies C
then $p$ satisfies A, E, or F (Property 1).
Consequently, it is not possible for
$p$ to satisfy B or D,
and $\non p$ to satisfy C.
\item % 2
If $p$ satisfies D, E, or F,
and $\non p$ satisfies C
then $p$ satisfies F (Property 2).
Consequently, it is not possible for
$p$ to satisfy D or E,
and $\non p$ to satisfy C.
\item % 3
If $p$ satisfies B or D and $\non p$ satisfies B or D
we have a contradiction.
That is,   it is not possible for
$p$ to satisfy B or D,
and $\non p$ to satisfy B or D.
\end{enumerate}

These effects apply for $p$ a positive or negative literal.

\ignore{
F is kind of obvious

A is consistent with all letters for much the same reason as F,
because the formulation of (2.2) in the $\pl$ inference rule treats both cases the same way.

C \& C and E \& E are both obvious

B \& E:
p -> p; => p;
~p => ~p
$\lambda \non p \notin P$
$\pl p \in P$

D \& E
=> p
~p => ~p
$\lambda \non p \notin P$
$\pl p \in P$
}

\begin{figure}
\begin{center}
%\verb|      | $\neg p$ \\
\begin{tabular}{lr||l|c|c|c|c|c|}
           \multicolumn{8}{c}{~~~$\neg p$}  \\ \\ 
           &          &   A  &   B      &   C     &   D       &   E      &   F     \\ \cline{2-8} 
           & ~A~  & \p   & \p        & \p       & \p        & \p        & \p      \\ \cline{2-8}
           & ~B~  & \p   & \np{3}  & \np{1} & \np{3}  & \p       & \p      \\ \cline{2-8}
\multirow{2}{*}{$p~~~$} & ~C~  & \p   & \np{1} &  \p      & \np{1}   & \np{2}   & \p    \\ \cline{2-8}
           & ~D~  & \p   & \np{3}  & \np{1} & \np{3} & \p         & \p     \\ \cline{2-8}
           & ~E~  & \p   & \p        & \np{2} & \p       & \p         & \p      \\ \cline{2-8}
           & ~F~  & \p   & \p        & \p       & \p       & \p          & \p     \\ \cline{2-8}
\end{tabular}
\end{center}
\caption{Table of all combinations of conclusions for $p$ and $\neg p$ in $\DL(\pl)$,
indicating whether or not the combination is possible, and, if not, why not.}
\label{fig:table}
\end{figure}

% Combinations that are possible in $\DL(\pl)$ but not in $\DL(\partial)$:
% AB, AD, EB, ED

In the table in Figure \ref{fig:table} we display the possible combinations of
conclusions for a proposition $p$ and its negation $\neg p$.
The table is symmetric across the leading diagonal,
since the treatment of literals in defeasible logic 
is independent of the polarity of the literal.
Those combinations which are possible are displayed as \p.
Those combinations which are not possible
are displayed as \np{i}, where $i$ is the property number in Proposition \ref{prop:props}
that implies that they are impossible.

For the possible combinations,
a sample theory can be exhibited by combining the sample theories for each letter (for $p$ and $\non p$, respectively).
We leave this for the reader to verify.

It is now straightforward to compare this table, for $\DL(\pl)$, 
with the table in \cite{table,TOCL01} for $\DL(\partial)$.
Every combination that is possible for $\DL(\partial)$ is also possible for $\DL(\pl)$.
Thus, for any defeasible theory $D$, and each proposition $p$, 
we can identify which combination of conclusions $\DL(\partial)$ entails
and simulate that behaviour with the sample theory for that combination for $p$.
This completes the proof of the first part of Theorem \ref{thm:pl_v_dl}.
Thus we have

%\setcounter{temp}{\value{theorem}}
%\setcounter{theorem}{\ref{thm:pl_v_dl}}
%\addtocounter{theorem}{-1}

\begin{theorem}  % \label{thm:pl_v_dl}
$\DL(\partial)$ is less expressive than $\DL(\pl)$ when there are no additions.
More specifically, 
\begin{itemize}
\item
every defeasible theory in $\DL(\partial)$ can be simulated by a defeasible theory in $\DL(\pl)$
\item
there is a defeasible theory $D$ whose consequences in $\DL(\pl)$
cannot be expressed by any defeasible theory in $\DL(\partial)$
\end{itemize}
\end{theorem}
\skipit{
\begin{proof}
The proof of the first part is established by the preceding work in this section.
For the second part,
let $D$ consist of 
\[
\begin{array}{lrcl}
        &             & \Rightarrow & \phantom{\neg} p \\
        & \neg p  & \rightarrow &  \neg p \\
\end{array}
\]
with empty superiority relation.

Then $+\lambda p$ and $+\pl p$ are consequences of $D$, as is $-\Delta p$, while $-\Delta \neg p $ is not a consequence.

We now show that $\DL(\partial)$ cannot simulate this theory.
That is, for no defeasible theory $D'$ 
are both $+\partial p$ and $-\Delta p$ consequences, and $-\Delta \neg p $ not a consequence.
Suppose $+\partial p$ is a consequence of $D'$.
Then either (1) $+\Delta p$ or (2) $-\Delta \neg p$ must be consequences of $D'$,
from the inference condition for $+\partial$.
But these contradict the requirements on the $\Delta$-consequences of $D'$.
Thus $\DL(\partial)$ is unable to simulate the consequences of $D$ under $\DL(\pl)$.
% Notice that $\DL(\partial)$ is unable to simulate $\Delta$ and $\lambda$.
%What about if we only have + consequences?  then it can simulate.
\end{proof}
}

%\setcounter{theorem}{\value{temp}}

% Notice that $\DL(\partial)$ is unable to simulate $\Delta$ and $\lambda$.
%What about if we only have + consequences?  then it can simulate.

It is interesting to note that the comparison of tables, identifies several different theories that
might be used to establish the second part of Theorem \ref{thm:pl_v_dl}.
However, they all have a similar structure: 
a literal $q$ is defeasibly provable, despite a loop for $\non q$.

%\setcounter{temp}{\value{theorem}}
%\setcounter{theorem}{\ref{thm:dl_v_pl}}
%\addtocounter{theorem}{-1}

\begin{theorem}   % \label{thm:dl_v_pl}
$\DL(\pl)$ is not more expressive than $\DL(\partial)$
with respect to addition of rules.
\end{theorem}
\skipit{
\begin{proof}
\ignore{
Consider the defeasible theory $D$
\[
\begin{array}{lrcl}
r_1: &        & \Rightarrow & \phantom{\neg} q \\
\end{array}
\]
and the rule addition
\[
\begin{array}{lrcl}
s: & \neg q & \rightarrow & \neg q \\
\end{array}
\]
The consequences of $D$ in $\DL(\partial)$ are
$+\partial q$,  $-\partial \neg q$, $\MD(q)$ and $\MD(\neg q)$ (and negative consequences for any other literal).
The only consequence of $D+A$ in $\DL(\partial)$ related to $q$ and $\neg q$ is
$\MD{q}$.

if we simulate $D$ as
$a => q$,
$=> a$,
$\neg q -> \neg a$
then $D$ and $D+A$ in $\DL(\pl)$ have exactly the same consequences (when restricted to language $q$).
Same when $A'$ is $\neg q$.

}

Let $D$ be the empty defeasible theory $(\emptyset,\emptyset,\emptyset)$.
The consequences of $D$ in $\DL(\partial)$ are
$-\partial q$,  $-\partial \neg q$, $\MD{q}$ and $\MD{\neg q}$ for every proposition $q$.
Suppose, to achieve a contradiction, that $D'$ is a simulation of $D$ wrt addition of rules in $\DL(\pl)$.
Then 
% $-\pl q$,  $-\pl \neg q$, 
$\MD{q}$ and $\MD{\neg q}$ are consequences of $D'$.
We will consider several additions to $D$.

For the addition $A_1$
\[
\begin{array}{lrcl}
r_1: &        & \Rightarrow & \phantom{\neg} q \\
\end{array}
\]
$D+A_1$ has consequences $+\partial q$ and $-\partial \neg q$.

Then $+\pl q$ is a consequence of $D'+A_1$.
It is straightforward that $\MD{q}$ and $\MD{\neg q}$ are consequences of $D'+A_1$.
\ignore{
We cannot have $+\Delta q \in P_\Delta$ for $D'+A_1$:
given $A$, this could only occur if $+\Delta q$ is a consequence of $D'$,
which contradicts the assumption that $D'$ is a simulation of $D$.
Similarly, we cannot have $+\Delta \neg q \in P_\Delta$.
}
It follows, from the $\pl$ inference rule, that (2.3.1) or (2.3.2) holds for each rule $s$ for $\neg q$ in $D'+A_1$.

Note that $+\lambda q \in P_\lambda$ for $D'+A_1$, by Proposition \ref{prop:contain}. 
Now, if $+\lambda \alpha \notin P_\lambda$ for $D'+A_1$, for some $\alpha$,
then also
$+\lambda \alpha \notin P_\lambda$ for $D'$,
because $P_\lambda$ monotonically increases with the addition of $+\lambda$ consequences.
Furthermore,
$r_1$ is not superior to any rule in $D'$ (by definition of modular addition).
Consequently, (2.3.1) or (2.3.2) holds for each rule $s$ for $\neg q$ in $D'$.

For addition $A_2$
\[
\begin{array}{lrcl}
r_2: &        & \Rightarrow & \phantom{\neg} q \\
s_2: & \neg q & \rightarrow & \neg q \\
\end{array}
\]
The only consequence of $D+A_2$ in $\DL(\partial)$ related to $q$ and $\neg q$ is
$\MD{q}$, and hence this is the only consequence of $D'+A_2$ in $\DL(\pl)$ related to $q$ and $\neg q$ .

As we saw from $A_1$, (2.3.1) or (2.3.2) holds for each rule $s$ for $\neg q$ in $D'$.
(2.3.2) cannot apply to $s_2$, by definition of modular addition.
If (2.3.1) applies to $s_2$ in $D'+A_2$ then
$+\lambda \neg q \notin P_\lambda$ for $D'+A_2$.    
Because neither $+\Delta q$ nor $+\Delta \neg q$ appear in $P_\Delta$ for $D'+A_2$,
this can only hold if every rule for $\neg q$ in $D'+A_2$
contains a body literal $\alpha$ such that $+\lambda \alpha \notin P_\lambda$ for $D'+A_2$,
by the $+\lambda$ inference rule.

Now consider the application of the $+\pl$ inference rule to prove $+\pl q$ in  $D'+A_2$.
(2.1) is satisfied by $r_2$ and (2.2) is satisfied.
For every rule for $\neg q$ in $D'+A_2$, (2.3.1) is satisfied, as shown in the previous paragraph.
Hence $+\pl q$ is a consequence of $D'+A_2$.
However, this contradicts the supposed simulation.
Thus there is no $D'$ that simulates $D$ wrt addition of rules in $\DL(\pl)$.

\ignore{
\finish{Isn't this a contradiction right here???? do we need $A_3$???}

For addition $A_3$ we use a new proposition $p$.
\[
\begin{array}{lrcl}
r_3: &        & \Rightarrow & \phantom{\neg} p \\
s_3: & \neg q & \Rightarrow & \neg p \\
\end{array}
\]
Consequences of $D+A_3$ in $\DL(\partial)$  include $+\partial p$
(because $s_3$ is not applicable) and $\MD{p}$ and $\MD{\neg p}$.
However, we cannot infer $+\pl p$ from $D'+A_3$.

Consider the application of the $\pl$ inference rule to infer $+\pl p$.
First, $\MD{p}$ and $\MD{\neg p}$ are in $P_\Delta$
and so (1) does not hold and (2.2) does hold.
(2.1) holds via rule $r_3$.
However, 
(2.3.1) does not apply to $s_3$ because $+\lambda \neg q \in P_\lambda$, and
(2.3.2) cannot apply to $s_3$, by definition of modular addition.
Thus $+\pl p$ is not a consequence of $D'+A_3$.
Hence $D'$ is not a simulation of $D$ wrt addition of rules.
}
\end{proof}
}